\documentclass[letterpaper, USenglish, numberwithinsect]{article}
\usepackage[margin=1in]{geometry}
\usepackage{style}

\graphicspath{{./figures/}}

\title{Sweeping Orders for Simplicial Complex Reconstruction}

\author{
  Tim Ophelders \thanks{\texttt{t.a.e.ophelders@uu.nl} \\\hspace*{1.8em}Department of Information and Computing Sciences, Utrecht University \\
 \hspace*{1.8em}Department of Mathematics and Computer Science, TU Eindhoven}
  \and
  Anna Schenfisch \thanks{\texttt{schenf@kth.se} \\ \hspace*{1.8em}Department of Mathematics, KTH Royal Institute of Technology}
}

\newcommand{\SSeq}{\mathit{SO}}

\begin{document}

\date{}
\maketitle
\begin{abstract}
    Standard sweep algorithms require an order of discrete points in Euclidean
    space, and rely on the property that, at a given point, all points in the
    halfspace below come earlier in this order. 
    We are motivated by the problem of reconstructing a graph in $\mathbb{R}^d$ from vertex locations and degree information, which was addressed using standard sweep algorithms by Fasy et al.
    We generalize this to the reconstruction of general simplicial complexes.
    As our main ingredient, we introduce a generalized \emph{sweeping order} on~$i$-simplices, maintaining the property that, at a given~$i$-simplex $\sigma$, all $(i+1)$-dimensional cofaces of $\sigma$ in
    the halfspace below $\sigma$ have an $i$-dimensional face that appeared
    earlier in the order (``below'' with respect to some direction
    perpendicular to $\sigma$).
    
    We then go on to incorporate computing such sweeping orders to reconstruct an unknown simplicial complex $K$, starting with only its vertex locations, i.e., its $0$-skeleton.
    Specifically, once we have found the $i$-skeleton of $K$, we compute a sweeping order for the $i$-simplices, and use it to reconstruct the $(i+1)$-skeleton of~$K$ by querying the \emph{indegree}, 
    or the number of $(i+1)$-simplices incident to and below a given $i$-simplex.
    In addition to generalizing the algorithm by Fasy et al. to simplicial complexes, we improve upon the running time of their central subroutine of radially finding edges above a vertex.
\end{abstract}

\section{Introduction}
    Since their introduction
    in~\cite{shamos1976geometric} to compute or count intersections of geometric
    objects, sweep algorithms have
    become ubiquitous in computational geometry, see, e.g.,~\cite{domiter2008sweep, vzalik2005efficient, fortune1986sweepline, ferreira2013parallel, lukavc2014sweep, edelsbrunner1986topologically, vzalik2023geometric, shih2005geometric}.
    For problems in~$\R^d$, such algorithms typically use a $(d-1)$-dimensional hyperplane as their \emph{sweeping object}, swept in some constant direction through the space.
    Sweep algorithms maintain solutions to subproblems associated with the region that was swept thus far.
    Throughout a sweep, there is a discrete set of locations, called \emph{sweep events}, where the solution to the subproblem may change combinatorially.
    For efficient sweep algorithms, this solution can be derived efficiently from solutions to the previous subproblems.

    Suppose we want to sweep, not through zero-dimensional sweep events, but
    through higher-dimensional sweep events, such as $i$-dimensional simplices. Unless these higher-dimensional sweep events are arranged in parallel, a single direction of sweep no longer suffices. How should we order these generalized sweep events to imitate fundamental properties of standard sweep algorithms?
    Not all properties of standard sweep algorithms are generalizable. As
    illustrated in~\cref{fig:star}, unlike the standard setting of zero-dimensional sweep events, we cannot necessarily
    find some ``initial'' $i$-dimensional sweep event that appears below all others. We instead turn to generalizing another fundamental property of standard sweeps, namely, that we can find some initial $i$-dimensional sweep event with no $(i+1)$-dimensional structures below it. Furthermore, the~$(i+1)$-dimensional structures below subsequent sweep events are above some previous sweep event. 

    \begin{figure}[h!]
        \centering
        \includegraphics{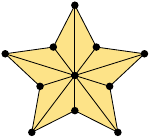}
        \caption{Here, every edge $e$ has other edges in the halfspaces on
        either side of $e$. This general higher-dimensional phenomenon contrasts
        the special zero-dimensional case; in every direction, we can find at
        least one vertex with no other vertices below it. However, for each
        dimension of simplex, there \emph{do} exist simplices with no
        \emph{cofacets} in a halfspace below it, which we will see is true
        in general.}
        \label{fig:star}
    \end{figure}
    
    Our motivation for considering higher-dimensional sweep events is generalizing~\cite{graphsearch},
    which sweeps up through a given set of vertices to reconstruct an unknown edge set.
    At each vertex $v$ in the sweep, they use a query called \emph{indegree} to count edges incident to $v$ in a particular halfspace.
    Using a radial search around $v$, they
    identify all edges incident to and above $v$. Crucially, they maintain that all edges with a vertex below the
    sweep height have already been identified.
    Indegree has a higher-dimensional generalization, namely, the count of
    cofacets of a given simplex in a particular halfspace. 
    Generalizing~\cite{graphsearch} to simplicial complex reconstruction
    as well as improving the runtime of the radial search method is a main result of this paper.

    While we consider indegree as a given function in this paper,
    it is possible to compute indegree from other geometric or topological information.
    Indegree is computed in~\cite{graphsearch, fasy2024faithful, belton2019reconstructing}
    using verbose persistence diagrams, a tool that records homological changes
    in a directional filtration.
    Indeed, simplicial complex reconstruction has close ties to
    faithful discretizations of \emph{directional transforms}, which are sets of
    (topological) descriptors corresponding to height/lower-star filtrations in directions that completely characterize the shape. 
    
    At the time of~\cite{graphsearch}, there was no method to compute
    higher-dimensional indegree using
    topological descriptors.
    However,~\cite{fasy2024faithful} developed exactly this tool, using an
    inclusion-exclusion type argument.
    But fundamentally, 
    as noted in~\cite{graphsearch}, a sweep of the vertex set is no longer feasible,
    since ``radially ordering higher dimensional simplices [around a common vertex]
    is not well-defined, and this issue prevents the methods [...] from
    being immediately transferable [to general simplicial complexes].'' 
    If the sweeping plane eventually contains some simplex~$\sigma$,
    it \emph{is} possible to radially order the cofacets of~$\sigma$ around $\sigma$. However, in general, no single direction is
    perpendicular to all simplices of some fixed positive dimension.
    We could try moving the sweeping plane non-linearly
    to contain each $i$-simplex, but in what order? An arbitrary order does 
    maintain that all~$(i+1)$-simplices below the sweeping plane have already
    been identified.

    In this paper, we show how to compute a \emph{sweeping order}, which conceptually 
    moves the sweeping plane so that all the properties necessary for the
    algorithm in~\cite{graphsearch} are maintained. This sweeping order is the
    key ingredient that allows us to generalize this edge
    reconstruction method to simplicial complexes.
    
\subparagraph{Outline}
    In \cref{sec:preliminaries}, we overview foundational topics
    and notation. We define our main tool in \cref{sec:sweeping}---a
    \emph{sweeping order} of a set of $i$-simplices---and show how such an order
    can be computed, even for degenerate simplicial complexes. 
    The edge
    reconstruction method of~\cite{graphsearch} uses a search through radially
    ordered \emph{candidate edges} incident to a given vertex. We generalize the
    notion of candidates to cofacets of a higher-dimensional simplex in
    \cref{sec:candidates}, and show how to order them radially around such a
    simplex in \cref{sec:candidateordering}. We determine which candidates are
    actually in the unknown simplicial complex~$K$ using a radial search that uses
    {indegree queries}.
   Finally, in \cref{sec:reconstruct}, we extend the methods of~\cite{graphsearch} to reconstruct the
    $(i+1)$-simplices of $K$ via indegree queries, given all lower-dimensional simplices. We can then iteratively reconstruct all of $K$ from just its vertex set
    and indegree queries. Omitted proofs and algorithms are in
    Appendix~\ref{append:proofs}.

\section{Preliminaries}
    \label{sec:preliminaries}
    In this section, we define our most extensively used terms, and refer the reader to~\cite{edelsbrunner2022computational, hatcher} for further information.
    We start with simplices and simplicial complexes.
    \begin{definition}[Simplex]\label{def:simplex}
        An \emph{abstract $i$-simplex} $\sigma$ is a set of $i+1$ elements, called vertices.
        The dimension of $\sigma$ is then $i$, denoted $\dim(\sigma)$. 
        If $\sigma$ and $\tau$ are abstract simplices and $\sigma\subseteq \tau$, 
        we call $\sigma$ a \emph{face} of $\tau$ and $\tau$ a \emph{coface} of $\sigma$. 
        If additionally $\dim(\sigma) = \dim(\tau) -1$, we call $\sigma$ a \emph{facet} of $\tau$ and~$\tau$ a \emph{cofacet} of $\sigma$.
        An \emph{$i$-simplex in $\R^d$} is an abstract $i$-simplex where each vertex maps to a distinct point in $\R^d$.
        We geometrically interpret a simplex as the convex hull of these points.
    \end{definition}
    
    \begin{definition}[Simplicial Complex]\label{def:complex}
        An abstract simplicial complex $K$ is a set of abstract simplices, 
        such that if $\sigma \in K$ and $\rho \subseteq \sigma$, then $\rho \in K$.
        A \emph{simplicial complex in $\R^d$} consists of an abstract simplicial complex, where each vertex is mapped injectively to a point in $\R^d$.
        Geometrically, we think of the simplicial complex as the union of convex hulls of its simplices.
    \end{definition}
    
    Our work always interprets simplicial complexes geometrically, so going
    forward, we may be somewhat less precise with terminology. Note that, so
    far, we allow for simplicial complexes $K$ with degenerate simplices and
    have no restrictions on how $K$ lies in $\R^d$, beyond injectivity of its
    vertex set. \cref{def:embedded} allows us to refine our view and discuss
    various types of simplicial complexes based on how their simplices
    intersect, see also \cref{fig:embedded-locally-neither}.
    
    \begin{definition}[Embedded and locally injective]\label{def:embedded}
        We call a pair of simplices $\sigma$ and $\sigma'$ an \emph{injective pair} if the intersection of their convex hulls is either empty or the convex hull of a common face (i.e., a simplex $\rho$ that is a face of $\sigma$ and $\sigma'$).
        We call a simplicial complex in~$\R^d$ \emph{embedded} if all pairs of simplices are injective pairs.
        We call a simplicial complex in~$\R^d$ \emph{locally injective} if all pairs of simplices that have a common face are injective pairs.
    \end{definition}

    \begin{figure}
        \centering
        \includegraphics{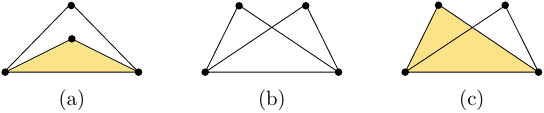}
        \caption{Examples of complexes in $\R^2$ that are (a) both locally injective and embedded, (b) locally injective but not embedded, and (c) neither locally injective nor embedded.}
        \label{fig:embedded-locally-neither}
    \end{figure}
    
    Let $K$ be a simplicial complex in $\R^d$.
    Let $\dim(K)$ denote the maximum dimension over all simplices in~$K$.
    We denote the $i$-skeleton of $K$
    by $K_i$, the number of $i$-simplices by $n_i$, and the total number of simplices in~$K$ by $|K|$. 
    For a simplex $\sigma \in K$, let $\cofacets{\sigma}$ denote all cofacets of~$\sigma$ and $\aff(\sigma)$ denote the affine hull of $\sigma$.
    
    The set of all unit vectors in $\R^d$ is parameterized by the unit $(d-1)$-sphere, 
    denoted~$\sph^{d-1}$, and a unit vector is called a \emph{direction}. 
    Denote by $\prp_\sigma\subseteq\sph^{d-1}$ the set of directions perpendicular to~$\sigma$ 
    and note that, if $\dim(\aff(\sigma)) = i' \leq d-1$, then $\prp_\sigma$ is a $(d-i'-1)$-sphere. 
    With respect to some $s\in\prp_\sigma$, all points in $\sigma$ have the same height, 
    which we refer to as the~\emph{$s$-coordinate} of~$\sigma$. 
    We may abuse dot product notation and write $s \cdot \sigma$ to denote this $s$-coordinate.
    For~$s\in\prp_\sigma$, the set $\downcofacets{s}{\sigma}\subseteq \cofacets{\sigma}$, consists of the cofacets~$\sigma \cup \{v\}$  for which $s\cdot v < s\cdot\sigma$, i.e., the vertex $v$ lies in the open halfspace below~$\sigma$ with respect to direction~$s$.
    We may also write $\cldowncofacets{s}{\sigma}$ if $s \cdot v \leq s \cdot \sigma$.
    Similarly, the set~$\upcofacets{s}{\sigma}$, consists of those cofacets~$\sigma \cup \{v\}$ for which we have $s\cdot v\geq s\cdot\sigma$.
    Regardless of our choice of $s \in \perp_\sigma$, we have $\downcofacets{s}{\sigma} \cup \upcofacets{s}{\sigma} = \cofacets{\sigma}$.
            
    Conceptually, our algorithms involve ordering points by rotating a hyperplane around some central space. In the following definition, we ensure that we make a consistent choice of normal direction to associate with each point we encounter.

 \begin{definition}[$\gamma$-Normal of $p$]\label{def:angle}
        Consider a unit circle $S\subset\R^d$, angularly parameterized by $\gamma\colon[0,2\pi)\to S$ and centered at some point $c\in\R^d$.
        For a point $p \in \R^d$, let $q\in \R^d$ be the orthogonal projection of $p$ onto the plane containing $\gamma$.
        If $q\neq c$, let~$\alpha$ be the unique angle
        such that the ray from $c$ through $\gamma(\alpha - \pi/2 \mathrm{~mod~} 2\pi)$ passes through $q$.
        We call $\gamma(\alpha)$ the \emph{$\gamma$-normal of $p$ relative to $c$}.
        If $q=c$, we define the $\gamma$-normal of $p$ to be~$\gamma(0)$.
        
        \begin{figure}[b]
            \centering
            \includegraphics{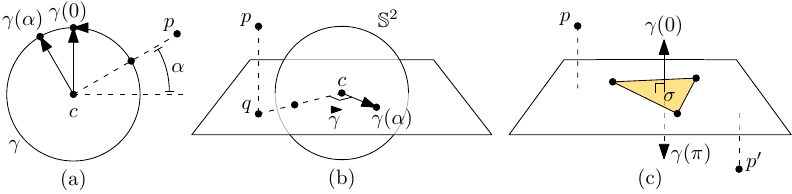}
            \caption{(a)-(b) $\gamma(\alpha)$ is the $\gamma$-normal of $p$. (c)
            since $\sigma$ has codimension one with the ambient space, any point
            above (or below) $\sigma$ has a $\gamma$-normal of $\gamma(0)$ (or
            $\gamma(\pi)$, respectively).}
            \label{fig:relative_normal}
        \end{figure}
        Now consider a simplex $\sigma\subset\R^d$ with
        $\dim(\aff(\sigma))<d-1$, let $\gamma$ be an angularly parameterized
        circle of directions all perpendicular to $\sigma$, and let $p$ be a
        point in $\R^d$.
        Then, the $\gamma$-normal of $p$ relative to any point $c$ in $\aff(\sigma)$ is the same, so we unambiguously define the \emph{$\gamma$-normal of $p$ relative to $\sigma$} to be the $\gamma$-normal of $p$ relative to any such point $c$.
        See~\cref{fig:relative_normal}~(a) and~(b).
    
        Finally, consider a simplex $\sigma\subset\R^d$ with $\dim(\aff(\sigma)) = d-1$.
        Let $\gamma\colon[0,2\pi)\to\sph^{d-1}$ be an angularly parameterized circle of directions so that $\gamma(0)$ is one of the two (antipodal) directions perpendicular to $\sigma$.
        For any point $p\in\R^d$, define the \emph{$\gamma$-normal of~$p$ relative to $\sigma$} to be $\gamma(0)$ if the $\gamma(0)$-coordinate of~$p$ is at least that of $\sigma$, and $\gamma(\pi)$ otherwise.
        See~\cref{fig:relative_normal}~(c).
    \end{definition}
    
    If $v_1 \neq v_2$ are vertices so that $\aff(\sigma \cup \{v_1\}) = \aff(\sigma\cup \{v_2\})$, 
    but $v_1$ and $v_2$ are on ``opposite sides'' of $\sigma$, and if the $\gamma$-normal of $v_1$ relative to $\sigma$ is $\gamma(\alpha)$, then the normal for $v_2$ is $-\gamma(\alpha)$.
    Note that if the $\gamma$-normal of some point relative to $\sigma$ is
    $\gamma(\alpha)$, then $\alpha$ is the angle at which, rotating a halfspace
    around $\sigma$, $p$ lies on the boundary of this halfspace and for which
    this halfspace has the exterior normal $\gamma(\alpha)$.

\section{Sweeping Orders}\label{sec:sweeping}
    The first goal of this paper is to show how to compute an order on simplices
    (along with a corresponding list of directions perpendicular to the
    simplices) that emulates properties characteristic of sweepline
    algorithms with discrete points or vertices as events.

    While an interesting problem by itself, we are motivated by \emph{simplicial complex reconstruction}, where an unknown simplicial complex is ``learned'' through a series of queries.\footnote{In practice, such queries arise from a mix of geometric and topological information about subcomplexes, for instance, from persistence diagrams. We discuss connections to topological data analysis and directional transforms in \cref{append:discretization}.} In particular, suppose that for a (known) $i$-simplex $\sigma$ and direction $s \in \perp_\sigma$, as long as we have already found $\downcofacets{s}{\sigma}$, we are able to learn $\upcofacets{s}{\sigma}$. Then our process of learning \emph{all} $(i+1)$-simplices should begin with some $\sigma_1$ and $s_1$ pair such that $\downcofacets{s_1}{\sigma_1}$ is empty, and the query condition is trivially satisfied. But which $i$-simplex should we process next? In other words, in what order and for what directions $s$ must we process $i$-simplices so that, when we query to learn $\upcofacets{s}{\sigma}$, we are guaranteed to have already learned $\downcofacets{s}{\sigma}$? We observe that if every $(i+1)$-simplex of $\downcofacets{s}{\sigma}$ is also a cofacet of other $i$-simplices that we had already successfully processed, then $\downcofacets{s}{\sigma}$ is known; defining and computing an order and associated directions for which this is possible is the aim of this section.
    
    Specifically, we introduce an order such that, 
    for a given $i$-simplex $\sigma$ and a corresponding direction~$s \in \perp_\sigma$, 
    all cofacets of $\sigma$ contained in $\downcofacets{s}{\sigma}$ are cofacets of some prior $i$-simplex in the order.
    Note that this generalizes a property of the usual sweep through a geometric
    graph, where vertices are sweep events. We formalize this definition below.
    
        \begin{definition}[Sweeping Order]\label{def:feasible_query_order}
            A \emph{sweeping order} of $K_i$ is any sequence $\SSeq_i:=((\sigma_j,s_j))_{j=1}^{n_i}$ that satisfies the following three properties.
            \begin{enumerate}
                \item Each $s_j$ is a direction perpendicular to its paired $i$-simplex, $\sigma_j$.
                \label{prop:perpendicular}
                \item Each $i$-simplex of $K$ appears exactly once in $\SSeq_i$.
                \label{prop:once}
                \item For any $i$-simplex $\sigma_j$, any cofacet in $\downcofacets{s_j}{\sigma_j}$ is also a cofacet of some $\sigma_h$, for $h < j$.
                \label{prop:halfspace}
            \end{enumerate}
        \end{definition}
        
\subsection{Computing a Sweeping Order}
    \label{sec:sweeping_computing}
    In this section, we show how to compute a sweeping order $\SSeq_i$ of $K_i$.
    For $K_0$, we simply pick an arbitrary direction $s$, order vertices by their $s$-coordinate (breaking ties arbitrarily), and output $(v,s)$ for each vertex $v$ in that order. See \cref{alg:sweep}.
    
   To compute a sweeping order for $K_i$ with $i>0$, we assume that we already know a sweeping order $\SSeq_{i-1}$ of $K_{i-1}$ (for example because we have computed $\SSeq_{i-1}$ recursively).
    Then $\SSeq_i$ is the sequence of pairs output by \Call{Order}{$K_i, \SSeq_{i-1}$} in \cref{alg:sweep}.
    We generally write $\rho$, $\sigma$, and $\tau$ to mean a simplex of dimension~$i-1$,~$i$, and $i+1$, respectively.
    We consider the $(i-1)$-simplices $\rho$ in the order prescribed by $\SSeq_{i-1}$.
    First, suppose $\dim(\aff(\rho)) \leq d-2$.
    We radially order the cofacets $\sigma$ of $\rho$ that have not yet been output.
    Specifically, if $\SSeq_{i-1}$ pairs $\rho$ with direction $s$, then the radial
    order of its cofacets is based on a parameterized circle
    $\gamma_\rho\colon [0,2\pi) \to \sph^{d-1}$ of
    directions, rotating around~$\rho$, starting at $s$.
    Each direction~$\gamma_\rho(\alpha)$ is the exterior normal of
    the unique halfspace whose boundary contains~$\rho$.
    For each cofacet $\sigma=\rho\cup\{v\}$ that has not yet been output, we consider the angle~$\alpha_v$ for which $\sigma$ enters this halfspace.
    Specifically, we consider~$\alpha_v$, the angle such that~$\gamma_\rho(\alpha_v)$ is the 
     $\gamma_\rho$-normal of $v$ relative to~$\rho$.
    We then output these cofacets~$\rho\cup\{v\}$ in increasing order based on $\alpha_v$, breaking ties arbitrarily, and pair them with the corresponding direction $\gamma_\rho(\alpha_v)$.
    \Cref{append:sweeping-order} walks through an example of \cref{alg:sweep}.

    If we encounter a simplex $\rho$ whose affine hull has dimension $d-1$, there are only
    two (antipodal) directions perpendicular to it; in this case, we simply choose
    an $\sph^1$ containing these two directions, so that all its cofaces have one of
    two possible $\gamma$-normals relative to $\rho$.
    For clearer exposition, we encompass these two cases in the following definition.
    \begin{definition}[Maximally Perpendicular Circle]
        Let $K$ be a simplicial complex in $\R^d$ and let $\sigma \subseteq K$ be an
        $i$-simplex with $\dim(\aff(\sigma)) < d$. We say that a circle of directions, $\gamma_\sigma$, is
        \emph{maximally perpendicular (to $\sigma$)}, if  
        \begin{enumerate}
            \item When $\dim(\aff(\sigma)) < d-1$, we have $\gamma_\sigma \subseteq \perp_\sigma$, or;
            \item When $\dim(\aff(\sigma)) = d-1$, we have $\perp_\sigma \subseteq \gamma_\sigma$. 
        \end{enumerate}
    \end{definition}
    That is, generally, $\gamma_\sigma$ ``rotates around'' $\sigma$, except in the case
    that $\sigma$ only has two directions perpendicular to it, in which case, these
    directions are contained in $\gamma_\sigma$. In either case, $\gamma_\sigma$ is ``as
    perpendicular'' to $\sigma$ as a circle of directions possibly can be.

    \begin{algorithm}[h!]
        \caption{Computing sweeping orders: $\SSeq_0$, and $\SSeq_i$ given $\SSeq_{i-1}$ and $K_i$.}
        \label{alg:sweep}
        \begin{algorithmic}[1]
            \Procedure{Order}{complex $K_0$}
                \State $s\gets$ arbitrary direction
                \label{algln:sweep:vertdir}
                \For{vertex $v$ of $K_0$ sorted increasingly by $s$-coordinate, breaking ties arbitrarily}
                \label{algln:sweep:vertexbegin}
                    \State \Output{$(v,s)$}
                    \label{algln:sweep:vertex}
                \EndFor
                    \label{algln:sweep:vertexend}
            \EndProcedure
            \Statex
            \Procedure{Order}{complex $K_i$, sweeping order $\SSeq_{i-1}$ of $K_{i-1}$}
                \For{$(\rho,s)$ in $\SSeq_{i-1}$}\Comment{$s$ is a direction perpendicular to the $(i-1)$-simplex $\rho$}
                \label{algln:sweep:forbegin}
                    \State $\gamma_\rho \gets$ \begin{varwidth}[t]{32.5em}\par{circle $\gamma_\rho\colon [0,2\pi)\to \sph^{d-1}$ of directions maximally perpendicular to~$\rho$, where the angle between $s$ and $\gamma_\rho(\alpha)$ is $\alpha$, so $\gamma_\rho(0)=s$ and $\gamma_\rho(\pi)=-s$}\end{varwidth}
                    \label{algln:sweep:gamma}
                    \State $U_\rho \gets \{ \text{vertex $v$ of $K_0$} \mid \text{$\rho\cup\{v\}$ is an $i$-simplex of $K_i$ that was not yet output}\}$
                    \label{algln:sweep:urho}
                    \For{$v\in U_\rho$}
                    \label{algln:sweep:vertexset}
                        \State $\alpha_v \gets \alpha$ such that $\gamma_\rho(\alpha)$ is the $\gamma_\rho$-normal to $v$ relative to $\rho$.
                        \label{algln:sweep:alpha}
                    \EndFor
                    \For{$v\in U_\rho$ sorted increasingly by $\alpha_v$, breaking ties arbitrarily}
                    \label{algln:sweep:sort}
                        \State \Output{$(\rho\cup\{v\},\gamma_\rho(\alpha_v))$}
             \label{algln:sweep:forendfinal}
                    \EndFor
                \EndFor
            \EndProcedure
        \end{algorithmic}
    \end{algorithm}

    \begin{remark}
        In $\R^0$ or $\R^1$, we cannot find \emph{any} $\sph^1$ of
        directions in the ambient space. In this very particular case, we invite the
        reader to imagine $\R^0$ or $\R^1$ along with any simplicial complex it contains as
        being included into $\R^2$, e.g., along the $x$-axis. Then we are able to
        proceed just as we do in the general case.
        Thus, we assume that $d>1$ for ease of exposition.
    \end{remark}
    
    Since our method requires the use of maximally perpendicular circles of directions, and a sweeping order pairs simplices with directions perpendicular to them, we require every simplex to have at least some $\sph^0$ of perpendicular directions.
    We formalize this in \cref{ass:perp}, which we henceforth assume is satisfied by $K$. %
    
    \begin{assumption}[General Position for Enough Perpendiculars]
    \label{ass:perp}
        Let $K$ be a simplicial complex in~$\R^d$.
        For every simplex $\sigma \in K_i$, we
        require that $\dim(\aff(\sigma)) \leq d-1$.
    \end{assumption}
    
    Note that this is quite a lenient condition, and does not prevent
    degeneracies. \Cref{ass:perp} automatically holds when $K_i$
    is at most $(d-1)$-dimensional.
    Additionally, \cref{ass:perp} allows for simplices of higher dimension as long as their affine hulls are sufficiently low-dimensional.
    
    Proceeding with this assumption, \Call{Order}{$K_i, \SSeq_{i-1}$},
    \cref{alg:sweep}, takes as input $K_i$ and $\SSeq_{i-1}$, a sweeping order
    for $K_{i-1}$, and outputs a sequence of $i$-simplices and directions. The
    main result of this section is \cref{thm:sweep}, which says
    \Call{Order}{$K_i, \SSeq_{i-1}$} is a sweeping order
    for $K_{i}$.
    We first show that it satisfies Properties~\ref{prop:perpendicular} and \ref{prop:halfspace} of \cref{def:feasible_query_order}.
    
    \begin{lemma}[Directions are Perpendicular to their Paired Simplices]\label{lem:alg_sweeping_perpendicular}
    \label{lem:perpendicular}
        Let $0 \leq i \leq \dim(K)-1$.
        If $i=0$, let $\SSeq_i=\Call{Order}{K_i}$.
        If $i>0$, let $\SSeq_i=\Call{Order}{K_i, \SSeq_{i-1}}$ for some sweeping order $\SSeq_{i-1}$.
        For all elements $(\sigma, s_\sigma)\in\SSeq_i$, the direction $s_\sigma$ is perpendicular to $\sigma$. 
        That is, the output of \cref{alg:sweep} satisfies \propref{perpendicular} of \cref{def:feasible_query_order}.
    \end{lemma}
    \begin{proof}
         Consider an arbitrary $(\sigma, s_\sigma)\in\SSeq_i$.
         If $i=0$, then $\sigma$ is a vertex, then any direction is perpendicular to $\sigma$, including $s_\sigma$.
         So consider the case $i>0$.
         Then $\sigma = \rho \cup \{v\}$ for some~$(i-1)$-simplex $\rho$ and vertex $v$, 
         where $(\rho, s_\rho)$ is an index for the loop in \alglnref{sweep:forbegin}, 
         and~$v$ is an element of $U_p$ (\alglnref{sweep:vertexset}). 
         On \alglnref{sweep:gamma}, we find the angle $\alpha_v$ such that $\gamma_\rho(\alpha_v)$ is the~$\gamma_\rho$-normal of $v$ relative to $\rho$.
         By \cref{ass:perp}, we have $\dim(\aff(\rho)) \leq d-1$.
         
         If $\dim(\aff(\rho)) < d-1$, then $\gamma_\rho$ is perpendicular to $\rho$, so the $\gamma_\rho$-normal $\gamma_\rho(\alpha_v)$ is well-defined and hence perpendicular to $\sigma$.
         If $\dim(\aff(\rho)) = d-1$, then we must have $\aff(\sigma) = \aff(\rho)$, 
         otherwise $\dim(\aff(\sigma))$ would be greater than $d-1$, violating \cref{ass:perp}. 
         Then the~$\gamma_\rho$-normal of $v$ relative to $\rho$ is $\gamma_\rho(0)=s_\rho$, which is perpendicular $\rho$ and hence also to~$\sigma$.
    \end{proof}
    
    \begin{lemma}[Halfspace Property]
    \label{lem:halfspace}
        Suppose that \Call{Order}{$K_i, \SSeq_{i-1}$} $=((\sigma_j, s_j))_{j=1}^{n_i}$. Then, for any $1 \leq j \leq n_i$, for each simplex $\tau$ of $\downcofacets{s_j}{\sigma_j}$, $\tau$ is a cofacet of some $\sigma_h$ with $h<j$. That is, \cref{alg:sweep} satisfies \propref{halfspace} of \cref{def:feasible_query_order}.
    \end{lemma}
    
    \begin{proof}
        \begin{figure}[b]
            \centering
            \includegraphics{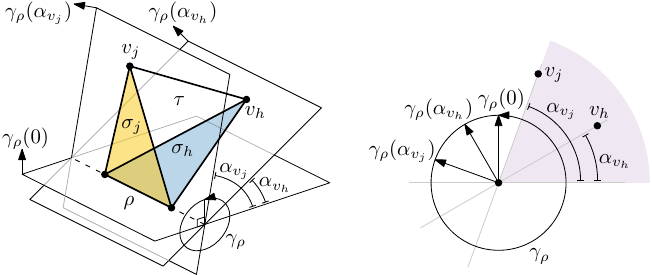}
            \caption{
            For some $\sigma_j=\rho\cup\{v_j\}$, we consider a cofacet $\tau=\rho\cup\{v_j,v_h\}$ (unshaded) for which the vertex $v_h$ lies below $\sigma_j$ with respect to the direction $\gamma_\rho(\alpha_{v_j})$. 
            The simplex $\sigma_h=\rho\cup\{v_h\}$ is a cofacet of $\rho$ that also has $\tau$ as a cofacet. 
            The simplex $\rho$ has a perpendicular circle of directions $\gamma_\rho$.
            }
            \label{fig:halfspace}
        \end{figure}
        When $i=0$, on \alglnref{sweep:vertex}, we output vertices ordered by their height with respect to some direction $s$, breaking ties arbitrarily. 
        Then any edge of a vertex $v_j$ whose other endpoint~$v_h$ lies in the open halfspace below $v_j$ with respect to $s$ is trivially also an edge of~$v_h$, and $v_h$ comes before $v_j$ in the ordering induced by~$s$, so the property is satisfied.
        
        Next, consider the case $i>0$ and suppose that \Call{Order}{$K_i, \SSeq_{i-1}$}, outputs $(\sigma_j,s_j)$ in the iteration $(\rho,s)$.
        Let $v_j$ be the vertex such that $\sigma_j=\rho\cup\{v_j\}$.
        We show for any simplex~$\tau=\sigma_j\cup\{v_h\}\in\downcofacets{s_j}{\sigma_j}$, that the simplex
        $\sigma_h=\rho\cup\{v_h\}$ satisfies the claim.
        
        Since $\sigma_h$ is a cofacet of $\rho$, and the iteration of $(\rho,s)$ outputs all cofacets of $\rho$ that have not yet been output,~$\sigma_h$ is output either in iteration of $(\rho,s)$, or in a previous iteration.
        If~$\sigma_h$ was output in a previous iteration, then $\sigma_h$ comes before $\sigma_j$ in the order, so $h<j$.
        Consider the remaining case that $\sigma_h$ is output during iteration $(\rho,s)$. 
        Since $(\rho, s)$ is an element of a sweeping order for $K_{i-1}$, it
        satisfies \propref{halfspace} of \cref{def:feasible_query_order}, so both
        $\sigma_j$ and $\sigma_h$ must be elements of $\upcofacets{s}{\rho}$, otherwise
        they would have been output in a previous iteration.
        Then we have $s\cdot v_j\geq s\cdot\rho$ and $s\cdot v_h\geq s\cdot\rho$, where $s=\gamma_\rho(0)$.
        
        Let~$\gamma_\rho(\alpha_{v_j})$ and $\gamma_\rho(\alpha_{v_h})$ denote $\gamma_\rho$-normals of $v_j$ and $v_h$ relative to $\rho$, respectively. 
        Because~$\sigma_j\cup\{v_h\}\in\downcofacets{s_j}{\sigma_j}$, and since $s_j = \gamma_\rho(\alpha_{v_j})$, 
        we have $\gamma_\rho(\alpha_{v_j}) \cdot v_h < \gamma_\rho(\alpha_{v_j}) \cdot \sigma_j$.
        Then $v_h$ lies (as illustrated by the shaded sector in~\cref{fig:halfspace}) in the open halfspace that contains $\sigma_j$ in its boundary with exterior normal direction $\gamma_\rho(\alpha_{v_j})$, 
        but not in the open halfspace that contains~$\rho$ in its boundary with exterior normal direction $\gamma_\rho(0)$.
        This means that $0\leq \alpha_h < \alpha_j$, so $\sigma_h$ is output before $\sigma_j$ on \alglnref{sweep:sort}, and the claim is satisfied.
    \end{proof}

    \propref{once} of \cref{def:feasible_query_order} is satisfied trivially, and we arrive at our first main result.
    
    \begin{restatable}{theorem}{sweep}\label{thm:sweep}
        For $i=0$, if we choose $s$ to be the first standard basis vector, \cref{alg:sweep} computes a sweeping order \Call{Order}{$K_0$} for $K_0$ in $O(n_0 \log n_0)$ time. For $i>0$, \cref{alg:sweep} computes a sweeping order
        \Call{Order}{$K_i, \SSeq_{i-1}$} for $K_{i}$
        in $O(n_{i-1} d i \min\{d,i\} + n_i(i+d+\log n_0))$ time.
    \end{restatable}

\begin{proof}
    The output of \cref{alg:sweep} satisfies Properties~\ref{prop:perpendicular} and~\ref{prop:halfspace} of \cref{def:feasible_query_order} by Lemmas \ref{lem:perpendicular} and \ref{lem:halfspace}. It satisfies Property~\ref{prop:once} trivially, and thus, is a sweeping order for $K_i$. 

    In \alglnref{sweep:vertexbegin}, for $i=0$ and choosing $s$ to be the first standard basis vector, we compute the $s$-coordinates of all vertices in $O(n_0)$ time, and sort them in $O(n_0 \log n_0)$ time.
    Next, suppose $i>0$. The loop of \alglnref{sweep:forbegin} has $n_{i-1}$ iterations.
    On \alglnref{sweep:gamma}, we find $\gamma_\rho$ in $O(d i \min\{d,i\})$ time using Gaussian elimination and a QR-decomposition, similar to \cite[Algorithm 5]{fasy2024faithful}.
    We can compute the $U_\rho$ of \alglnref{sweep:urho} in $O(|\cofacets{\rho}|)$ time. 
    Over all $\rho$'s, this takes $O(i n_i)$ time.
    We compute all $\alpha_v$'s of \alglnref{sweep:alpha} for a particular $\rho$ in $O(d|\cofacets{\rho}|)$ time. 
    Over all $\rho$'s, this takes $O(d n_i)$ time.
    In \alglnref{sweep:sort}, each $\rho$ has at most $n_0$ cofacets, so we sort~$U_\rho$ in $O(|\cofacets{\rho}|\log |\cofacets{\rho}|)=O(|\cofacets{\rho}|\log n_0)$ time.
    Over all $\rho$'s, this sorting takes $O(\sum_\rho |\cofacets{\rho}|\log n_0)=O(n_i\log n_0)$ time.
    The total running time is~$O(n_{i-1} d i \min\{d,i\} + i n_i + d n_i + n_i\log n_0)=O(n_{i-1} d i \min\{d,i\} + n_i(i+d+\log n_0))$.
\end{proof}

    For our application in \cref{sec:reconstruct}, it will be useful to retain the information of the maximally perpendicular circles used on Line~\ref{algln:sweep:gamma}, and call a sweeping order that additionally records this information \emph{circle-reporting}.
    We can modify \cref{alg:sweep} to output a circle-reporting sweeping order while maintaining the same asymptotic running time, and denote this modified version of \cref{alg:sweep} by \textsc{Order$^\circ$}.

\section{Candidate Simplices}
    \label{sec:candidates}
    Given $K_i$, there may be vertices defining an $(i+1)$-simplex, $\tau$, that simply cannot be an
    $(i+1)$-simplex of $K$.
    For instance, if not all facets of $\tau$ are in $K_i$, then we know that~$\tau$ is not part of $K$ and we do not need to consider~$\tau$ in our search.
    Moreover, if we know $K$ satisfies extra properties, e.g., that it is embedded, and adding $\tau$ to~$K_i$ would violate that property, there is no reason to consider~$\tau$ in our search. 
    \Cref{def:candidate} specifies which potential simplices our search should
    consider, knowing that the to-be-reconstructed $K$ satisfies
    property~$X$. 
    Intuitively,
    a candidate cofacet of $\sigma$ is an~$(i+1)$-simplex that, at least,
    \emph{may} be contained in~$K$: determining if an
    $(i+1)$-simplex is or is not in $K$ is the main content
    of~\cref{sec:reconstruct}. 
    \begin{definition}[Candidate Vertices and Candidate Cofacets]\label{def:candidate}
         Let $K$ be a simplicial complex in~$\R^d$ with property $X$ (locally
         injective, embedded, no condition, etc.), and let $\sigma\subseteq K$ be an $i$-simplex
         for some $i < \dim(K)$. Suppose that there exists a vertex $v \in K_0$ such that
         the simplices in $K_i$ defined on $\sigma \cup \{v\}$ form the boundary of
         an $(i+1)$-simplex $\tau$ such that $K_i \cup \tau$ is a subcomplex of some
         simplicial complex with property $X$. We call $v$ a \emph{candidate vertex}
         of~$\sigma$. Furthermore, we call the $(i+1)$-simplex defined by $\sigma
         \cup \{v\}$ a \emph{candidate cofacet of $\sigma$}. 
    \end{definition}
        
    \cref{fig:candidates} shows all candidate simplices, for various properties of the simplicial complex $K$.
    We assume that the set of candidates of dimension $i+1$ can be determined once we have correctly reconstructed~$K_i$.
    We henceforth assume that these candidates are known to us, and denote the candidate cofacets of a simplex $\sigma$ by $\cand(\sigma)$.

    \begin{figure}[b]
        \centering
        \includegraphics{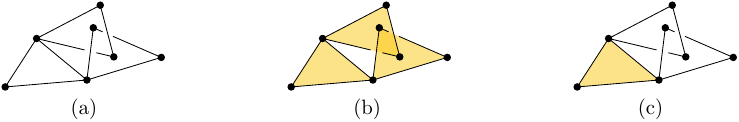}
        \caption{
        Suppose that~(a) is the one-skeleton of an otherwise unknown simplicial complex~$K$.
        If~$K$ is known to be locally injective, (b) depicts the set of candidate triangles.
        If~$K$ is known to be embedded, we instead get the set of candidate triangles shown in (c).
        } 
        \label{fig:candidates}
    \end{figure}
 
    Recall that computing a sweeping order used only mild assumptions on the
    underlying simplicial complex (\cref{ass:perp}), and allowed
    degeneracies. However, with our current goal of reconstruction, we introduce
    stricter assumptions on $K_i$ and its candidate simplices.
    
    \begin{assumption}[General Position for $(i+1)$-Reconstruction]
    \label{ass:reconstruction}
        Let $K$ be a simplicial complex in~$\R^d$. Suppose, for every $i$-simplex
        $\sigma$, the set $\cand(\sigma)$ of all candidate cofacets of~$\sigma$ along with $K_i$ is locally injective in~$\R^d$. Then we say that $K$ is in
        \emph{general position for $(i+1)$-reconstruction}.
    \end{assumption}
    \begin{observation}\label{obs:embedded_candidates}
        \Cref{ass:reconstruction} is automatically satisfied if $K$ is known to be embedded.
    \end{observation}

    \begin{figure}[t]
        \centering
        \includegraphics{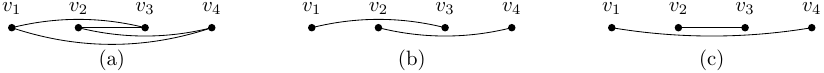}
        \caption{
        Suppose that (a) displays the set of candidate edges for $v_1, v_2, v_3$, and $v_4$, which violates \cref{ass:reconstruction}. Note that this can lead to non-reconstructible simplicial complexes, as (b) and (c) have the same indegree
        information for all directions.}
        \label{fig:assumption2}
    \end{figure}
    
    \Cref{fig:assumption2} justifies the need for \cref{ass:reconstruction} by giving simplicial complexes violating the assumption that are \emph{not} reconstructible via indegree queries.
    Note that a stronger property $X$ may lead to fewer candidate cofacets, and a weaker \cref{ass:reconstruction}.
    For a candidate cofacet~$\tau$, \cref{ass:reconstruction} implies that $\dim(\tau) = \dim(\aff(\tau))$.
    It also implies the following.
    
    \begin{restatable}{lemma}{twosides}
    \label{lem:twosides}
        Let $\sigma$ be an $i$-simplex of a simplicial complex $K$ in $\R^d$ where $i<d$.
        Let~$C_{i+1}$ denote the candidate cofacets of $\sigma$. If
        $K_i \cup C_{i+1}$ is locally injective, then only two candidate cofacets of
        $\sigma$ can share the same affine hull, and such candidates are separated by $\aff(\sigma)$.
    \end{restatable}
    
    \begin{proof}
         Let $\tau$ be a candidate cofacet of $\sigma$. Because $K_i \cup C_{i+1}$ is locally injective, the dimension of $\aff(\tau)$ is one more than $\aff(\sigma)$.
         The space $\aff(\sigma)$ is a
         separating hyperplane in $\aff(\tau)$, so there is a well-defined notion of
         being on a particular side of $\aff(\sigma)$ in $\aff(\tau)$.
         Let $\tau'$ be some other candidate cofacet of $\sigma$ and suppose, towards a contradiction, that $\aff(\tau) = \aff(\tau')$, and both
         $\tau$ and $\tau'$ lie on the same side of $\aff(\sigma)$.
         Then $\sigma$ is the largest common face of $\tau$ and $\tau'$, but their intersection contains more than $\sigma$, so they are not an injective pair, contradicting our assumption of local injectivity for $K_i \cup C_{i+1}$.
         Finally, since there are only two sides of $\aff(\sigma)$ in an $(i+1)$-dimensional
         hull, candidate cofacets of $\sigma$ that share
         a common affine hull can only come in pairs on opposite sides.
    \end{proof}

\section{Radially Ordering Candidate Cofacets}
    \label{sec:candidateordering}
    In this paper, we use circles of maximally perpendicular
    directions in two separate but related contexts. In \cref{sec:sweeping}, we
    found a sweeping order for $K_i$ by rotating around circles maximally
    perpendicular to each~$(i-1)$-simplex. In this section, we discuss
    a
    second type of maximally perpendicular circle, around which we rotate to order candidate cofacets of some central $i$-simplex in the reconstruction process.
     
    \begin{definition}[Candidate-Ordering Circle]
        Let $K$ be a simplicial complex in $\R^d$ and consider an $i$-simplex $\sigma
        \subseteq K$ for which there exists some parameterized circle of directions, $\gamma_\sigma\colon [0, 2\pi) \to \R^d$,
        that is maximally perpendicular to~$\sigma$.
        If every candidate vertex of $\sigma$ has a unique
        $\gamma_\sigma$-normal relative to $\sigma$, then we say
        $\gamma_\sigma$ is \emph{candidate-ordering}.
    \end{definition}
    
    Since orders from different angular parameterizations are simply cyclic permutations of each other, we generally do not specify parameterizations. 
    Towards establishing the existence of candidate-ordering circles, the following lemma shows, for a simplex $\sigma$ with~$\dim(\sigma) = i < d-1$,
    how to build a
    circle of directions $\gamma_\sigma$ so that any direction in it is
    perpendicular to at most two candidate cofacets of $\sigma$ simultaneously, and
    such candidate pairs share affine hulls.
    
    \begin{lemma}
    \label{lem:constructS}
        Let $K$ be a simplicial complex in $\R^d$ for $d \geq 2$, and let $\sigma
        \subseteq K$ be an $i$-simplex with $i < d-1$. Let $C_{i+1}$ denote the
        candidate cofacets of $\sigma$. If $K_i \cup C_{i+1}$ is locally injective,
        then a circle of directions $\gamma_\sigma$ maximally perpendicular to $\sigma$ exists so
        that whenever a direction of $\gamma_\sigma$ is normal to two candidate cofacets~$\tau$ and $\tau'$ of
        $\sigma$, the affine hulls of $\tau$ and $\tau'$ are equal.
    \end{lemma}
    \begin{proof}
        We show the existence of $\gamma_\sigma$ constructively. First, suppose
        $i=d-2$. Then $\gamma_\sigma = \perp_\sigma$ is a maximally perpendicular
        circle, and the affine hulls of candidate cofacets of $\sigma$ are~$(d-1)$-planes normal to some direction in $\gamma_\sigma$. By
        \cref{lem:twosides}, these planes contain at most two candidate cofacets and
        the claim is satisfied.
    
        Next, consider $i < d-2$. Let $\tau$ and $\tau'$ be candidate
        cofacets of $\sigma$. Then $\aff(\tau \cup \tau')$ is~$(i+2)$-dimensional
        (if $\aff(\tau) \neq \aff(\tau')$), or else is $(i+1)$-dimensional (if
        $\aff(\tau) = \aff(\tau')$). In particular, the arrangement of the planes
        $\aff(\tau \cup \tau')$ for every pair of candidate cofacets $\tau$
        and~$\tau'$ is less than $d$-dimensional. Then we can choose a set of
        $d-i-2$
        points not contained in this arrangement, which we denote $P$, in such a way
        so that $\aff(\sigma \cup \{P\})$ is $(d-2)$-dimensional. Let
        $\gamma_\sigma$ be the unique $\sph^1$ of directions (maximally)
        perpendicular to~$\aff(\sigma \cup \{P\})$.
        
        Now, suppose there is some $s \in \gamma_\sigma$ perpendicular to two candidate cofacets, $\sigma\cup \{v\}$ and~$\sigma \cup \{v'\}$. That is, for all $p \in \sigma \cup \{P\}$, we have
        $s \cdot v=s \cdot v' = s \cdot p$. Then $\sigma\cup \{P\} \cup \{v, v'\}$
        is a set of $(d-1) + 2 = d+1$ points that all lie on the same $(d-1)$-plane (normal to $s$). But by construction,
        $\dim(\aff(\sigma \cup \{P\})) = d-2$, so it must be that $\aff(\tau) =
        \aff(\tau')$.
    \end{proof}
    
    See \cref{fig:candidate-ordering} for a low-dimensional illustration of the construction described above.

    \begin{figure}[h]
        \centering
        \includegraphics{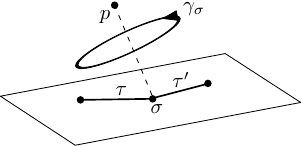}
        \caption{
        Here, $\sigma$ is a vertex in $\R^3$. Then we choose $P= \{p\}$ to avoid
        the affine hull of the single pair of (candidate) cofacets of
        $\sigma$. The resulting circle $\gamma_\sigma$ has no direction
        orthogonal to both $\tau$ and $\tau'$; this would only occur if $\tau$
        and $\tau'$ were colinear. The arrow on $\gamma_\sigma$ hints at our
        eventual aim; $\gamma_\sigma$ defines a total radial order on the
        candidate cofacets of $\sigma$.}
        \label{fig:candidate-ordering}
    \end{figure}

    We combine Lemmas~\ref{lem:twosides} and~\ref{lem:constructS} to show that a candidate-ordering circle always exists.

    \begin{restatable}{lemma}{order}
    \label{lem:order}
        Let $K$ be a simplicial complex, let $\sigma$ be an $i$-simplex for $i \leq d-1$, and suppose
        that $K$ satisfies \cref{ass:reconstruction}. Then there exists a
        candidate-ordering $\gamma_\sigma$.
    \end{restatable}
    
    \begin{proof}
        We proceed with a constructive argument. If $i=d-1$, then by
        \cref{lem:twosides}, $\sigma$ can have at most two candidate cofacets,
        contained on opposite sides of $\sigma$.
        Then, for any circle $\gamma_\sigma$ of directions maximally perpendicular to~$\sigma$, the $\gamma_\sigma$-normal relative to $\sigma$ of one candidate cofacet is $\gamma_\sigma(0)$, and that of the other is~$\gamma_\sigma(\pi)$; these are distinct, so the claim holds.
    
        Next, suppose $i < d-1$, and suppose that $\gamma_\sigma \subseteq \perp_\sigma$ is
        the circle maximally perpendicular to $\sigma$, constructed as in the proof of \cref{lem:constructS}. 
        Choose an angular parametrization of $\gamma_\sigma$ by $[0, 2\pi)$.
        By \cref{lem:constructS}, whenever a direction $\gamma_\sigma(\alpha)$ is
        orthogonal to multiple candidate cofacets, they share affine hulls. By
        \cref{lem:twosides}, such candidate cofacets lie on opposite sides of
        $\sigma$, and can only come in pairs. Thus, such candidate cofacet pairs have opposite (and distinct) $\gamma_\sigma$-normals relative to $\sigma$.
        Thus, $\gamma_\sigma$ is a candidate-ordering circle for $\sigma$.
    \end{proof}

       In \cref{sec:reconstruct}, we consider sweeping orders where, for each pair $(\sigma, s)$, the direction~$s$ is part of some candidate-ordering circle around $\sigma$.
    We show that this requirement is satisfiable:

    \begin{restatable}{lemma}{orderingordering}
    \label{lem:orderingordering}
        Let $i \leq d$ and let $K$ be a simplicial complex in $\R^d$ that is in general position for $(i'+1)$-reconstruction, for all $i' \leq i \leq d-1$.
        There exists a sweeping order $((\sigma_j, s_j))_{j=1}^{n_i}$ where, for all $1 \leq j \leq n_i$, 
        there is a candidate-ordering circle around $\sigma_j$ that contains $s_j$.
    \end{restatable}
    
    As an immediate consequence of the proof of this \cref{lem:orderingordering}, we see that computing this particular
    type of sweeping order via \cref{alg:sweep} is possible.

    \begin{restatable}{corollary}{orderingorderingcor}
        \label{cor:orderingordering}
        If we input a candidate-ordering-compatible sweeping order to \cref{alg:sweep}, the output is also candidate-ordering-compatible.
    \end{restatable}
    
    Thus, we can (and will) reuse the circles used to compute a sweeping order for $K_i$ to order candidates cofacets of $i$-simplices; recall that, to annotate a sweeping order with these circles, we call the adapted ``circle-reporting'' \cref{alg:sweep}, \textsc{Order$^\circ$}.

\section{Simplicial Complex Reconstruction}        
    \label{sec:reconstruct}
    In this section, we leverage the properties of our sweeping order to the task of simplicial complex reconstruction, thereby generalizing the edge reconstruction
    algorithm of~\cite{graphsearch}.
    Just as in~\cite{graphsearch}, we sweep through the~$i$-simplices, finding all cofacets above a given $i$-simplex at each step, maintaining that all cofacets below have already been found. 
    We iterate the process of reconstructing $K_{i+1}$ from $K_i$ until we reconstruct $K$. 
    Since a sweeping order for $K_{i+1}$ is computed from $K_i$, the orders we utilize are computed along the way.
    
\subsection{Finding Cofacets above a Single \texorpdfstring{$i$}{i}-Simplex}
\label{sec:upcofacets}
    First, we discuss our central subroutine; identifying all cofacets of and above
    a single $i$-simplex~$\sigma$, supposing that all cofacets of and below $\sigma$
    have already been found.
    Our main tool is the following function, which
    counts cofacets of a simplex in the closed halfspace below the simplex with respect
    to some direction.

    \begin{definition}[Indegree]
        \label{def:indegree}
        Given $\sigma$, some simplex of a simplicial complex $K$, and $s \in
        \perp_\sigma$, {\sc{Indeg}$(\sigma, s)$} returns the number of cofacets
        of $\sigma$ that have no vertex higher than $\sigma$ with respect to
        direction $s$. That is, ${\sc{Indeg}}(\sigma, s) = \lvert
        \cldowncofacets{s}{\sigma}\rvert$.
    \end{definition}
 
    For now, we assume {\sc{Indeg$(\sigma, s)$}} as a well-defined subroutine; we
    provide further discussion about its actual existing implementations and
    corresponding limitations in \cref{append:discretization}.

    \begin{figure}[b]
        \centering
        \includegraphics{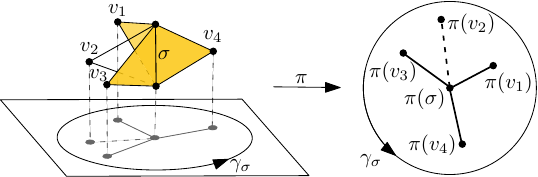}
        \caption{
        Projecting $\sigma$ and its candidate cofacets to the plane
        containing a candidate-ordering circle $\gamma_\sigma$ results in a star
        graph with the image of $\sigma$ at the center, and the order of the three candidate
        vertices in the image remains unchanged. Note that
        $\sigma \cup \{v_2\}$ is not an actual cofacet of $\sigma$; then, despite
        appearing as an edge in the projection, it does not contribute to an
        indegree count.
        }
        \label{fig:projection}
    \end{figure}

    We present an algorithm {\sc{FindUnfound}} (\cref{alg:findunfound}), inspired by the algorithm \textsc{UpEdges}~\cite[Alg 2]{graphsearch}, which addresses the specific case $i=0$.%
    \footnote{Using a single direction to mean ``up,'' an edge adjacent to and above a vertex $v$ will never be simultaneously strictly above its other endpoint. However, using the directions of our sweeping order as ``ups,'' an upper cofacet of an $i$-simplex $\sigma$ for $i>0$ may indeed be an upper cofacet of some other of its $i$-dimensional faces. 
    Thus, in higher-dimensional contexts, we may already
    know cofacets in the set~$\upcofacets{s}{\sigma}$; noting their presence and not ``re-finding'' them is more efficient than searching for everything in $\upcofacets{s}{\sigma}$. This is why we name \cref{alg:findunfound} {\sc{FindUnfound}} rather than, e.g., {\sc{UpCofacets}}.}
    Our algorithm improves upon the running time of {\sc{UpEdges}}, and simultaneously generalizes it to higher dimensions ($i>0$).%

    Given an $i$-simplex with $i<d-1$, and a candidate-ordering circle,
    the corresponding radial
    ordering of candidate cofacets behaves
    nearly identically to the radial ordering of edges incident to some central
    vertex. \Cref{fig:projection} highlights this connection.
    For the specific case~$i=d-1$, we can no
    longer rotate around the simplex, but a simple indegree check confirms the
    existence or absence of an upper cofacet.

    For the case $i < d-1$, we begin by radially sorting both the known cofacets
    of $\sigma$ (\textsc{Known}) and the candidate cofacets of $\sigma$ (\textsc{Cand}). Note that $\textsc{Known}$ is a subset of $\textsc{Cand}$.
    We denote the number of known cofacets of $\sigma$ that lie in the closed halfspace through $\sigma$ with exterior normal $s^*$ by $\textsc{\#Known}_{s^*}^\leq$.
    For each candidate cofacet, we store its $\gamma$-normal $s_c$ and its angle $\alpha_c$ along $\gamma$ in $O(d)$ time, and also precompute and store $\textsc{\#Known}_{s_c}^\leq$.
    As part of \cref{lem:findunfound}, we show how to do this in linear time.
    
    On \alglnref{findunfound:tobefound} we perform an initial indegree check in the $-s$ direction and subtract the number of already known cofacets $\#\textsc{Known}_{-s}^\leq$ in the upper halfspace.
    This counts how many cofacets of $\sigma$ are yet to be found, denoted \#\textsc{InitiallyUnfound}.
    Until this many cofacets have been found, we repeatedly perform a binary search for the next to-be-found cofacet of $\sigma$ (see \cref{fig:upcofacets}), which we then output, thus finding all cofacets of~$\sigma$.
    
    \begin{algorithm}
        \caption{Consider a simplex $\sigma\in K$, a direction $s \in \perp_\sigma$, a candidate-ordering circle $\gamma$ starting at $s$, and a list, \textsc{Known}, of vertices known to define cofacets of $\sigma$.
        If \textsc{Known} includes all vertices defining cofacets in $\downcofacets{s}{\sigma}$, then \Call{FindUnfound}{$\sigma,s,\gamma, \textsc{Known}$}
        finds all vertices of cofacets of $\sigma$ that were previously unknown (i.e., not in \textsc{Known}). }
        \label{alg:findunfound}
        \begin{algorithmic}[1]
        \Procedure{FindUnfound}{$\sigma, s, \gamma, \textsc{Known}$}
               \State $\textsc{Known} \gets$ $\textsc{Known}$ sorted increasingly by angle along $\gamma$
               \label{algln:findunfound:known}
               \State $\textsc{Cand} \gets$ the list of candidates of $\sigma$, sorted increasingly by their angle along $\gamma$
               \State for each candidate~$\textsc{Cand}[c]$, store its $\gamma$-normal $s_c=\gamma(\alpha_c)$ and corresponding angle $\alpha_c$
               \label{algln:findunfound:cand}
               \State $l \gets$ index of last candidate $\textsc{Cand}[c]$ with $\alpha_c\leq\pi$ (or $0$ if no such candidate exists)
                \label{algln:findunfound:lastuppercand}
               \State for each $\textsc{Cand}[c]$, precompute $\#\textsc{Known}^{\leq}_{s_c}$
               \label{algln:findunfound:precompute}
               \State \textsc{\#InitiallyUnfound} $\gets \Call{Indeg}{\sigma, -s} - \#\textsc{Known}_{-s}^{\leq}$
               \label{algln:findunfound:tobefound}
            \If{$\dim(\sigma) = d-1$ and $\#\textsc{InitiallyUnfound}>0$}
            \label{algln:findunfound:dminusoneif}
                \State \Output the first element of $\textsc{Cand}$
                \label{algln:findunfound:dminusone}
            \EndIf
            \If{$\dim(\sigma) < d-1$}
               \State \textsc{\#Found} $\gets 0$
               \label{algln:findunfound:foundinit}
               \While{$\textsc{\#Found}<\textsc{\#InitiallyUnfound}$}
               \label{algln:findunfound:outerwhilebegin}
                   \State $a \gets 1$
                   \State $b \gets l+1$
                   \While{$a+1<b$}
                   \label{algln:findunfound:innerwhilebegin}
                       \State $c \gets \lfloor{\frac{a+b}{2}}\rfloor-1$
                       \If{$\Call{Indeg}{\sigma, s_c} > \textsc{\#Found} + \# \textsc{Known}^{\leq}_{s_c}$} \Comment{not all of $\cldowncofacets{s_c}{\sigma}$ is found or known}
                       \label{algln:findunfound:innerifbegin}
                           \State $b \gets c+1$
                       \Else
                           \State $a \gets c+1$
                       \EndIf
                   \EndWhile
                   \label{algln:findunfound:innerwhileend}
                   \State $\textsc{\#Found} \gets \textsc{\#Found}+1$
                   \label{algln:findunfound:foundplusone}
                   \State \Output $\textsc{Cand}[a]$
               \EndWhile
               \label{algln:findunfound:outerwhileend}
            \EndIf
        \EndProcedure
        \end{algorithmic}
    \end{algorithm}
    
    \begin{figure}
	    \centering
		\includegraphics{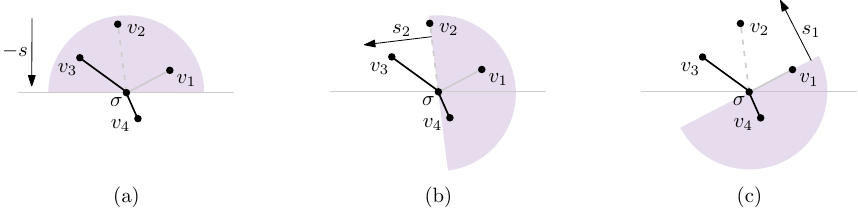}
		\caption{
        The first steps of {\sc{FindUnfound}}, where $\textsc{Known} = [v_3, v_4]$ and where $\textsc{Cand} = [v_1, v_2, v_3, v_4]$, out of which only $v_3$ does not form a cofacet with $\sigma$. We begin in (a) by computing $\#\textsc{InitiallyUnfound} = \textsc{Indeg}(-s, \sigma) - \textsc{Known}_{-s}^\leq = 2-1=1$, i.e., we know there is one cofacet of sigma ``to be found'' in our search space. In (b), we begin our binary search, finding $\textsc{Indeg}(\sigma, s_2) = 2 >  0+1 = \#\textsc{Found} - \#\textsc{Known}_{s_2}^\leq$. We then continue the binary search in (c), finding $\textsc{Indeg}(\sigma, s_1) =2 > 0+1 = \#\textsc{Found} - \#\textsc{Known}_{s_1}^\leq$. Updating $b=3$ to $b=2$ breaks the while condition, so we add one to $\#\textsc{Found}$ and output $\textsc{Cand}[1]= \sigma \cup \{v_1\}$.
        Since $\#\textsc{Found} = \#\textsc{InitiallyUnfound}$, we know we have found all cofacets of $\sigma$ and the algorithm ends.
        }
	\label{fig:upcofacets}
    \end{figure}

    \begin{lemma}
    \label{lem:findunfound}
        For $\sigma$, an $i$-simplex of a simplicial complex $K$, $s \in
        \perp_\sigma$, and $\gamma$, a candidate ordering circle for $\sigma$
        containing $s$, and \textsc{Known}, a list of vertices known to define
        cofacets of $\sigma$.
        Let $U = \#\textsc{InitiallyUnfound}$, and $I$ be the time to query the indegree of $\sigma$ in a given direction.
        If \textsc{Known} includes all vertices defining
        cofacets in $\downcofacets{s}{\sigma}$, then $\Call{FindUnfound}{\sigma, s, \gamma, \textsc{Known}}$ (\cref{alg:findunfound}) 
        \begin{enumerate}
            \item is correct: it finds all vertices of cofacets of $\sigma$ that were previously unknown (i.e., not in \textsc{Known});
            \item runs in $O(n_0d+I+(n_0+UI)\log n_0)$ time and performs $O(1+U\log n_0)$ indegree queries.
        \end{enumerate}
    \end{lemma}
    
    \begin{proof}        
        We start with the correctness.
        First, consider the case that $\dim(\sigma) = d-1$.
        If ${\#\textsc{InitiallyUnfound}=0}$, then no to-be-found cofacets lie above $\sigma$.
        By \cref{lem:twosides}, at most one candidate vertex $v$ may lie above $\sigma$.
        Therefore, if $\#\textsc{InitiallyUnfound}>0$, then $\#\textsc{InitiallyUnfound}=1$, and the only to-be-found cofacet of $\sigma$ is the first candidate, whose vertex we correctly return on \alglnref{findunfound:dminusone}.
        
        Next, consider the case that $\dim(\sigma) < d-1$.
        We claim that during the $j$-th iteration, the loop in \alglnrefRange{findunfound:outerwhilebegin}{findunfound:outerwhileend} outputs the $j$-th element of $\textsc{Cand} \setminus \textsc{Known}$ that is a cofacet of $\sigma$, which we refer to as the $j$-th ``to-be-found'' cofacet of $\sigma$, so that after $j$ iterations, the first $j$ to-be-found cofacets have been output.
        To establish this, we now argue that the loop of \alglnrefRange{findunfound:innerwhilebegin}{findunfound:innerwhileend} performs a binary search for the next to-be-found cofacet of~$\sigma$.
        
        Let $\rho$ be the next to-be-found cofacet of $\sigma$.
        Let $\gamma(\alpha_\rho)$ be the $\gamma$-normal of $\rho$.
        Because $\rho$ is unknown, we have $\alpha_\rho\leq\pi$.
        We maintain that $\textsc{Cand}[a, \dots,  b-1]$ contains $\rho$.
        The invariant holds initially, since $\textsc{Cand}[a, \dots,  b-1]$ contains all candidates $c$ with $\alpha_c\leq\pi$, and thus contains $\rho$.
        We show that the invariant is maintained.
        On \alglnref{findunfound:innerifbegin}, either $\rho\in\textsc{Cand}[a, \dots, c]$, or $\rho\in\textsc{Cand}[c+1, \dots, b-1]$.
        First consider the case that $\rho\in\textsc{Cand}[a, \dots, c]$. 
        Since $0\leq \alpha_\rho \leq \alpha_c \leq \pi$, we know $\rho$ lies in the halfspace below $\sigma$ with respect to $s_c$.
        So this halfspace contains at least one to-be-found cofacet of $\sigma$.
        Moreover, all candidates $\textsc{Cand}[f]$ previously found by the while-loop have $0\leq \alpha_f\leq\alpha_c$, and therefore also lie in this halfspace.
        Thus, $\rho$ is counted by $\textsc{Indeg}(\sigma, s_c)$.
        Since $\rho$ has not been found, nor was it known, we have $\Call{Indeg}{\sigma, s_c} > \textsc{\#Found} + \# \textsc{Known}^{\leq}_{s_c}$ on \alglnref{findunfound:innerifbegin}. 
        Therefore, the invariant is maintained if $\rho\in\textsc{Cand}[a, \dots, c]$.
        
        Now consider the other case that $\rho\in\textsc{Cand}[c+1, \dots, b-1]$.~Suppose, for the sake of contradiction, that $\Call{Indeg}{\sigma, s_c}>\textsc{\#Found}+\# \textsc{Known}^{\leq}_{s_c}$.
        Then there exists a to-be-found cofacet $\textsc{Cand}[u]$ of $\sigma$ with $\alpha_u\in[0,\alpha_c]\cup[\alpha_c+\pi,2\pi)$.
        Because all to-be-found cofacets have $\alpha_u\leq\pi$, we have $\alpha_u\in[0,\alpha_c]$, and so there is a to-be-found cofacet in $\textsc{Cand}[a, \dots, c]$, which contradicts that $\rho$ is the first to-be-found cofacet.
        We conclude that $\Call{Indeg}{\sigma, s_c}\leq\textsc{\#Found}+\# \textsc{Known}^{\leq}_{s_c}$, so the invariant is maintained.
        When the loop terminates, $\textsc{Cand}[a, \dots, b-1]$ contains only one to-be-found candidate, which by the invariant corresponds to $\rho$.
        This concludes the proof of correctness.
        
        Next, we analyze the runtime and number of indegree queries.
        Computing the $\gamma$-normals and angles of \textsc{Known} and \textsc{Cand}, and sorting them takes $O(n_0d+n_0\log n_0)$ time.
        Computing $l$ takes $O(\mathfrak{C})$ time.

        We show that \alglnref{findunfound:precompute} can be implemented to take $O(n_0)$ time, by radially sweeping along $\gamma$ using a parameter~$\alpha$.
        Specifically, consider the halfspace through $\sigma$ with exterior normal $\gamma(\alpha)$.
        Because the \textsc{Known} array is sorted by $\gamma$-normal, we can represent the range of elements of \textsc{Known} (and similarly of \textsc{Cand}) that lie in that halfspace using two pointers.
        As the halfspace rotates as $\alpha$ increases, we maintain these pointers, and when $\gamma(\alpha)=s_c$, we associate $\textsc{\#Known}_{s_c}^\leq$ to $\textsc{Cand}[c]$.
        The events of this radial sweep correspond to values of $\alpha$ for which elements of \textsc{Known} or \textsc{Cand} enter or leave the halfspace, so the number of events is $O(\textsc{\#Known}+\textsc{\#Cand})=O(n_0)$.
        Updating the pointers or associating $\textsc{\#Known}_{s_c}^\leq$ to $\textsc{Cand}[c]$ takes constant time per event.
        Therefore, \alglnref{findunfound:precompute} takes $O(n_0)$ time.

        \alglnref{findunfound:tobefound} takes $O(I+\textsc{\#Known})=O(I+n_0)$ time, using one indegree query and counting the number of elements of \textsc{Known} in the corresponding halfspace.
        \alglnrefRange{findunfound:dminusoneif}{findunfound:foundinit} take constant time.
        The while loop in \alglnrefRange{findunfound:outerwhilebegin}{findunfound:outerwhileend} takes $U$ iterations.
        The running time of each iteration is dominated by the binary search, which takes $O(I \log \textsc{\#Cand})=O(I\log n_0)$ time and $O(\log n_0)$ indegree queries.
        The total running time is therefore $O(n_0d+n_0\log n_0+I+UI\log n_0)$, and we use $O(1+U\log n_0)$ indegree queries.
    \end{proof}

    \begin{remark}[Comparison to UpEdges]
        Rather than using an indegree oracle, the algorithms of~\cite{graphsearch} give a specific method for computing indegree in $O(\log n_0 + \Pi)$ time, where $\Pi$ is the time required to compute a persistence diagram. That is, using the methods of~\cite{graphsearch}, $I = O(\log n_0 + \Pi)$.
        
        Then, \textsc{UpEdges} runs in $O((U \log n_0)(\log n_0 + d + I))$ time. However, it assumes that \textsc{Cand} and \textsc{Known} are already sorted.
        If we also make this assumption in \textsc{FindUnfound} and similarly restrict ourselves to the case $i=0$, our runtime becomes $O(I + (n_0 + U I) \log n_0)$ and is therefore an improvement on the running time of \textsc{UpEdges}.
        In large part, this is because each iteration of the binary search in \textsc{UpEdges} calls a subroutine (\textsc{SplitArc}), which in turn uses an additional binary search to find pointers that represent candidates in a particular wedge around~$\sigma$.
        By preprocessing the radially sorted list of candidates, we bypass the nested binary search.

        We observe that~\cite{graphsearch} relegates the cyclic ordering of candidates to a simultaneous precomputation made elsewhere in their method, using~\cite[Lems 1 and 2]{verma2011slow}. In view of the generality of our higher-dimensional methods, we have elected to retain our (slower) strategies of cyclic ordering, as we could only apply the methods of~\cite{verma2011slow} in the specific case $i=0$.
    \end{remark}

    In \cref{lem:findunfound}, having a known lower halfspace is crucial for the binary search, as it allows us to deduce the number of unfound cofacets with angles between $0$ and any particular query angle.
    In \cref{alg:reconstructnext}, we show that Property~\ref{prop:halfspace} of a sweeping order (\cref{def:feasible_query_order}) can be used to ensure a known lower halfspace.

\subsection{Reconstructing \texorpdfstring{$K$}{K}}
    From \cref{cor:orderingordering}, we can iteratively build up a
    candidate-ordering-compatible sweeping order for $K_i$, given~$K_{i-1}$.
    Processing $K_i$ in this order, \cref{alg:reconstructnext}
    reconstructs~$K_{i+1}$.
    
    \newcommand{\Found}{\mathit{Found}}
    \begin{algorithm}
            \caption{
            \Call{ReconstructNext}{$K_i, \SSeq\crep_i$}, for a candidate-ordering-compatible sweeping order, $\SSeq\crep_i$,
            computes the $(i+1)$-simplices of $K$.}
            \label{alg:reconstructnext}
            \begin{algorithmic}[1]
            \Procedure{ReconstructNext}{$K_i,\SSeq\crep_i$}
                \label{algln:reconstruct:empty}
                \For{$(\sigma,s,\gamma)$ in $\SSeq\crep_i$}
                    \State $\knbd[\sigma] \gets []$\Comment{the list of vertices known to define cofacets of $\sigma$}
                \EndFor
                \For{$(\sigma,s, \gamma)$ in $\SSeq\crep_i$}
                \label{algln:reconstruct:bigloopstart}
                    \For{$u \in $~\Call{FindUnfound}{$\sigma, s, \gamma, \knbd[\sigma]$}}
                    \label{algln:reconstruct:findunfound}
                        \State $\rho \gets \sigma \cup \{u\}$
                        \State \Output $\rho$ \Comment{newly found $(i+1)$-simplex}
                        \label{algln:reconstruct:output}
                        \For{$v \in \rho$}\Comment{record $\rho$ as a known cofacet of all its facets}
                        \label{algln:reconstruct:vertexstart}
                            \State append $v$ to $\knbd[\rho-v]$
                            \label{algln:reconstruct:vertexend}
                        \EndFor
                    \EndFor
                \EndFor
                \label{algln:reconstruct:bigloopend}
            \EndProcedure
            \end{algorithmic}
    \end{algorithm}
    
    \begin{lemma}
    \label{lem:reconstructnext}
        \cref{alg:reconstructnext} is correct, i.e., \Call{ReconstructNext}{$K_i, \SSeq_i$} outputs exactly the $i+1$-simplices of $K$. 
        Furthermore, it runs in $O\left(n_i\left( n_0 d + I + n_0\log n_0\right)+(I\log n_0+i)n_{i+1}\right)$ time, and performs $O(n_i + n_{i-1}\log n_0)$ indegree queries.
    \end{lemma}
    
    \begin{proof}
        Every output simplex is clearly an $(i+1)$-simplex of $K$.
        It remains to show that every $(i+1)$-simplex of $K$ is output.
        Let $(\sigma_j, s_j, \gamma_j)\in\SSeq\crep_i$ denote the triple processed in iteration~$j$ of the loop in 
        \alglnrefRange{reconstruct:bigloopstart}{reconstruct:bigloopend}. We claim that this loop satisfies the
        following invariant: after processing $\sigma_j$, the output includes
        $\cof(\sigma_j)$.
        Before entering the loop, this is vacuously true.
        
        Next, suppose that the invariant holds up until iteration $j$.
        That is, for each $h < j$, the output includes $\cofacets{\sigma_h}$.
        Now consider entering iteration~$j$.
        First, we claim that $\downcofacets{s_j}{\sigma_j}$ is already included in the output.
        Since the $i$-simplices are processed according to a sweeping order, 
        the direction $s_j$ is perpendicular to~$\sigma_j$ by \propref{perpendicular} of \cref{def:feasible_query_order},
        so $\downcofacets{s_j}{\sigma_j}$ is well-defined. 
        Furthermore, by \propref{halfspace} of \cref{def:feasible_query_order}, 
        all cofacets of $\sigma_j$ contained in $\downcofacets{s_j}{\sigma_j}$ 
        are cofacets of $i$-simplices that appeared earlier in the sweeping order. 
        By assumption, such $(i+1)$-simplices have already been output. 

        On \alglnrefRange{reconstruct:vertexstart}{reconstruct:vertexend}, we
        record all output $(i+1)$-simplices as cofacets for each of their
        respective facets in the set $\knbd$, specifically recording the vertex
        defining the cofacet-facet relation. In particular, each cofacet
        of~$\downcofacets{s_j}{\sigma_j}$ was output, and its defining vertex was added to $\knbd(\sigma_j)$ in the same iteration.
        Therefore, in iteration $j$, the vertices defining the set $\downcofacets{s_j}{\sigma_j}$ are included in $\knbd(\sigma_j)$.
        
        Then by \cref{lem:findunfound}, calling {\Call{FindUnfound}{$\sigma_j,
        s_j, \gamma_j, \knbd(\sigma)$}} on \alglnref{reconstruct:findunfound}, we 
        correctly return the list of vertices defining
        any cofacets of $\sigma_j$ that were previously unknown.
        Iterating over all
        vertices defining cofacets in this set and adding such cofacets to our
        output on \alglnref{reconstruct:output} means that, entering
        iteration~$j+1$, all cofaces of previously processed $i$-simplices have
        been output, maintaining the invariant.
        
        By \propref{once} of \cref{def:feasible_query_order}, 
        all $i$-simplices appear exactly once in $\SSeq\crep_i$, 
        so the loop iterates over all $i$-simplices in $K$. 
        Since each $(i+1)$-simplex is necessarily a cofacet of some~$i$-simplex, 
        when the loop terminates, all $(i+1)$-simplices have been found.

        Next, we analyze the running time.
        Let $U_\sigma$ denote the number of cofacets output in the iteration of \alglnref{reconstruct:findunfound} that processes $\sigma$.
        Note that because each simplex is output exactly once, we have $\sum_\sigma U_\sigma=n_{i+1}$.
        By \cref{lem:findunfound}, the call to \cref{alg:findunfound} in the iteration that processes $\sigma$ takes $O(n_0 d+I+(n_0 +U_\sigma I)\log n_0)$ time.
        Moreover, for each $(i+1)$-simplex $\rho$ that is output during this iteration, we iterate over the $i+2$ vertices $v$ of $\rho$ to insert $v$ into $\knbd[\rho-v]$.
        The sweeping order has $n_i$ elements, making the total running time
        \begin{align*}
             &\, O\left(\sum_{(\sigma,s,\gamma)\in\SSeq\crep_i}  \left( n_0 d + I + (n_0+U_\sigma I)\log n_0 + U_\sigma (i+2) \right)\right)\\
            =&\, O\left(\sum_{(\sigma,s,\gamma)\in\SSeq\crep_i}\left( n_0 d + I + n_0\log n_0+U_\sigma (I\log n_0+i) \right)\right)\\
            =&\, O\left(n_i\left( n_0 d + I + n_0\log n_0\right)+n_{i+1}(I\log n_0+i)\right).
        \end{align*}
        The algorithm uses $O(n_i+n_{i+1}\log n_0)$ indegree queries.
    \end{proof}
    
    Finally, we present \cref{alg:reconstructall}, which reconstructs $K$ given
    the vertex set $K_0$.
    \begin{algorithm}[H]
        \caption{
        \Call{ReconstructAll}{$K_0$} computes $K$, satisfying \cref{ass:reconstruction} for all $i<\dim(K)$.}
        \label{alg:reconstructall}
        \begin{algorithmic}[1]
            \Procedure{ReconstructAll}{$K_0$}
            \State $\SSeq\crep_0 \gets$ \Call{Order$^\circ$}{$K_0$}, computed to be candidate-ordering-compatible
            \label{algln:reconstructall:sweepinit}
            \State $i \gets 0$
            \label{algln:reconstructall:iequalszero}
            \While{$K_i$ has $i$-simplices} 
            \label{algln:reconstructall:whilebegin}
            \State $i \gets i+1$
            \label{algln:reconstructall:increase}
            \State $K_i \gets K_{i-1}~\cup$ \Call{ReconstructNext}{$K_{i-1}, \SSeq\crep_{i-1}$}
            \label{algln:reconstructall:reconstructnext}
            \State $\SSeq\crep_i \gets$ \Call{Order$^\circ$}{$K_i, \SSeq\crep_{i-1}$}, where on \alglnref{sweep:gamma}, $\gamma_\rho$ is chosen to be candidate-ordering%
            \label{algln:reconstructall:sweep}
            \EndWhile
            \label{algln:reconstructall:whileend}
            \State\Return $K_i$
        \EndProcedure
        \end{algorithmic}
    \end{algorithm}
    
    \begin{restatable}{theorem}{reconstructall}
        \label{thm:reconstructall}
        Let $K$ be a simplicial complex of dimension $\dim(K)=\kappa$ in $\R^d$ satisfying \cref{ass:reconstruction} for all~$i < \kappa$. 
        If we know $K_0$, then \cref{alg:reconstructall} reconstructs $K$, i.e., \Call{ReconstructAll}{$K_0$} $ = K$, 
        in $O\left(|K|\left( n_0 d + n_0\log n_0 + I\log n_0 + \kappa \right)+ |K_{\kappa-1}|d \kappa\min\{d,\kappa\}\right)$ time, 
        using $O(\vert K\vert \log n_0)$ indegree queries.
    \end{restatable}

\begin{proof}
    We claim the loop in
    \alglnrefRange{reconstructall:whilebegin}{reconstructall:whileend} satisfies
    the following invariant: entering iteration~$i$, we know $K_{i-1}$ as
    well as a sweeping order for $K_{i-1}$.
    Since the vertex set is taken as input, we know $K_0$ initially. In
    \alglnref{reconstructall:sweepinit}, we call \Call{Order}{$K_0$},
    which, by \cref{thm:sweep}, computes a sweeping order for $K_0$. Thus, the
    loop invariant holds before entering the loop. 
    Now suppose this invariant is true when entering iteration $j$ for some
    $j\geq 1$. That is, we know $K_{j-1}$ and a sweeping order for $K_{j-1}$,
    which we use as inputs in the call to {\sc{ReconstructNext}}
    on~\alglnref{reconstructall:reconstructnext}. By \cref{lem:reconstructnext},
    this outputs all $j$-simplices, which, combined with $K_{j-1}$, gives us $K_j$. Finally on~\alglnref{reconstructall:sweep}, we use
    {\sc{Order}} to compute a sweeping order for $K_{j}$, which,
    again, is correct by \cref{thm:sweep}. Thus, the loop invariant holds.
    Since $K$ is a finite simplicial complex, index $i$ will eventually reach
    $\dim(K)$. At this point, the $i$-skeleton $K_i$ equals the full simplicial
    complex $K$. When we set $i = \dim(K) +1$ in
    \alglnref{reconstructall:increase}, there are no $i$-simplices to find or
    order, so the output
    on~\alglnrefTwo{reconstructall:reconstructnext}{reconstructall:sweep} is
    empty. Thus, we exit the loop and correctly return $K_i = K$.

    Next, we analyze the running time of \cref{alg:reconstructall}.
    \alglnref{reconstructall:sweepinit} takes $O(n_0 \log n_0)$ time by \cref{thm:sweep}.
    The \alglnrefRange{reconstructall:iequalszero}{reconstructall:increase} each take constant time.
    The loop of \alglnrefRange{reconstructall:whilebegin}{reconstructall:whileend} is iterated $\dim(K)$ times.
    During iteration $i$, \alglnref{reconstructall:reconstructnext} takes $O(n_i\left( n_0 d + I + n_0\log n_0\right)+n_{i+1}(I\log n_0+i))$ time by \cref{lem:reconstructnext}, and \alglnref{reconstructall:sweep} takes ${O(n_{i-1} d i \min\{d,i\} + n_i(i+d+\log n_0))}$ time by \cref{thm:sweep}. Thus, the loop takes a total time of
    \begin{align*}
         &\,O\left(
            \sum_{i=1}^{\kappa} \left(n_i\left( n_0 d + I + n_0\log n_0\right)+n_{i+1}(I\log n_0+i)+ n_{i-1} d i \min\{d,i\} + n_i(i+d+\log n_0)\right)
        \right)\\
        =&\, O\left(
            \sum_{i=1}^{\kappa} n_i\left( n_0 d + I + n_0\log n_0 + I\log n_0+i + i+d+\log n_0\right)+ \sum_{i=0}^{\kappa-1} n_i d i \min\{d,i\}
        \right)\\
        =&\, O\left(
            \sum_{i=1}^{\kappa} n_i\left( n_0 d + n_0\log n_0 + I\log n_0 + \kappa \right)+ \sum_{i=0}^{\kappa-1} n_i d \kappa \min\{d,\kappa\}
        \right)\\
        =&\, O\left(
            |K|\left( n_0 d + n_0\log n_0 + I\log n_0 + \kappa \right)+ |K_{\kappa-1}|d \kappa\min\{d,\kappa\}
        \right).
    \end{align*}

    Since this loop dominates the running time, this is also the total running time of \cref{alg:reconstructall}.
    The number of indegree queries is $O(|K|\log n_0)$.
\end{proof}
    
    In the special case where the to-be-reconstructed complex $K$ is known to be embedded, there are no $(d+1)$-dimensional simplices, and we can terminate the loop of \alglnref{reconstructall:whilebegin} after constructing $K_d$ (and before constructing $\SSeq_d$).
    Whenever the algorithm computes $\SSeq_i$, we have $i<d$, so \cref{ass:perp} is automatically satisfied.
    By \cref{obs:embedded_candidates}, \cref{ass:reconstruction} is also satisfied, and \cref{thm:reconstruct_embedded} follows.
    \begin{theorem}\label{thm:reconstruct_embedded}
        Under the promise that $K$ is an embedded simplicial complex,
        there is an algorithm that can reconstruct $K$ using $O(|K|\log n_0)$ indegree queries in $O\left(|K|\left( n_0 d + n_0\log n_0 + I\log n_0 \right)+ |K_{\kappa-1}|d^3 \right)$ time.
    \end{theorem}
    
\section{Discussion}
    Our method for finding all cofacets above a given $i$-simplex, \textsc{FindUnfound} (\cref{alg:findunfound}) generalizes
    the corresponding algorithm
    \textsc{UpEdges} of~\cite[Alg 2]{graphsearch}. If we also assume candidate
    and known cofacets are presorted, as \textsc{UpEdges} does, we improve the running time by a log factor. Thus, for the specific case~${i=0}$, plugging in \textsc{FindUnfound} in the place of \textsc{UpEdges} improves the overall running time of the edge reconstruction method of~\cite{graphsearch}.
    
    Note that~\cite{fasy2024faithful} also gives a higher-dimensional simplex reconstruction method using indegree queries, but does not use any sort of radial search, and is therefore of quite a different nature than the method presented here. Furthermore,~\cite{fasy2024faithful} uses reconstruction only as a proof method for their main result, which, roughly, shows that a particular set of indegree queries is sufficient for reconstruction. They therefore begin with complete knowledge of $K$, allowing them to make efficient choices for query directions in a way that our framework does not allow. This is discussed in more detail in~\cref{append:faithful}. 
    
    We defined sweeping orders with the particular motivation of simplicial complex reconstruction, but we note that \cref{alg:sweep} could be adapted to order objects with a more general cell-structure. 
    In particular, we expect that \cref{alg:sweep} can be adapted to order faces of hyperplane arrangements.
    With such an order, we are curious what other traditional sweepline algorithms can be adapted to higher-dimensional generalizations.

\bibliography{references.bib}
    
\newpage
\appendix

    \section{Example of \Cref{alg:sweep}}
    \label{append:sweeping-order}
    In this appendix, we walk through an example of \cref{alg:sweep} for a particular simplicial complex, as illustrated in~\cref{fig:sweeping-order}.
    \begin{figure}
        \centering
        \includegraphics[page=2]{sweeping_order3}\hspace{1.6em}%
        \includegraphics[page=1]{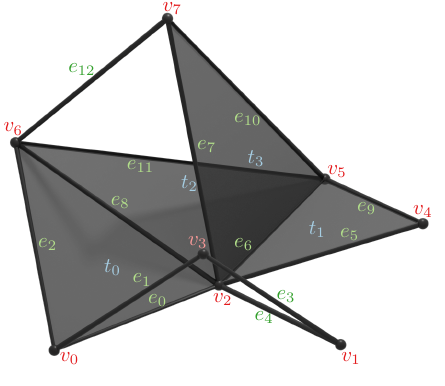}
        \caption{
            A simulation of \Call{Order}{$K_0$} with direction $s$, and \Call{Order}{$K_i,\SSeq_{i-1}$} for $i\in\{1,2\}$. 
            Some circles~$\gamma_\rho$ relevant to the order are shown. 
            Each $\gamma_{v_i}$ lies in a plane parallel to the page. 
            The circle $\gamma_{e_6}$ lies in a plane perpendicular to edge $e_6$.
            A dotted line connecting a simplex $\sigma$ to a facet $\rho$ indicates that $\rho$ is the facet that outputs $\sigma$.
            A solid line indicates the $\gamma_\rho$-normal with which some $\sigma$ is output.
            Indices correspond to the order in simplices are output.
        }
        \label{fig:sweeping-order}
    \end{figure}
    First, we compute \Call{Order}{$K_0$} with direction $s$ (shown as in the
    center of \cref{fig:sweeping-order}, pointing to the top of the page). This
    orders the seven vertices of $K_0$ by their height with respect to $s$. Explicitly, we compute the sweeping order $\SSeq_{i-1} = ((v_i, s))_{i=0}^7$.

    Next, we aim to compute \Call{Order}{$K_1, \SSeq_{0}$}. We begin the for loop on \alglnref{sweep:forbegin} with index $(v_0, s)$, and on \alglnref{sweep:gamma}, we choose $\gamma_{v_0}$ (shown in the left of the figure), which is trivially maximally perpendicular to $v_0$. Rotating around $\gamma_{v_0}$, we encounter edges in the order $e_0$, $e_1$, and $e_2$. Since none of these edges have been previously output, they are the first three entries in our sweeping order, each paired with the corresponding direction from $\gamma_{v_0}$ normal to that edge, which we will call $s_{e_0}, s_{e_1}$, and $s_{e_2}$, respectively.
    We then process $v_1$, choosing $\gamma_{v_1}$ as a maximally perpendicular circle, encountering and outputting $(e_3, s_{e_3})$ and $(e_4, s_{e_4})$ in that order.
    In the loop with index $(v_2, s)$, we encounter edges in the order $e_5, e_6, e_7, e_8$, and $e_0$. However, since $e_0$ has already appeared in the sweeping order, we simply output $e_5, e_6, e_7$, and $e_8$, along with their corresponding directions.
    The remaining iterations of the for loop are similar, and we obtain a sweeping order for $K_1$, denoted $\SSeq_1$.

    Finally, we compute \Call{Order}{$K_2, \SSeq_1$}. We begin the for loop on
    \alglnref{sweep:forbegin} with index $(e_0, s_{e_0})$. Only $t_0$ is
    incident to $e_0$, and corresponds to the single output for this
    iteration (the maximally perpendicular circle around $e_0$ is omitted to
    avoid visual clutter). Although we also encounter $t_0$ when rotating around
    $e_2$, we do not output it again. During iteration $(e_5, s_{e_5})$, we find
    $t_1$, and during iteration $(e_6, s_{e_6})$ we encounter $t_2$ and $t_3$.
  
  If $K$ is in $\R^3$, then each edge has a unique circle of directions
  (maximally) perpendicular to it; for $\R^d$ with $d > 3$, we may choose one.
  Regardless of the ambient dimension, one such circle is illustrated in the
  figure as $\gamma_{e_6}$.

    We emphasize that computing a sweeping order for $K_i$ only requires
    knowledge of the $i$-skeleton and a sweeping order for $K_{i-1}$. In
    particular, when computing sweeping orders for reconstruction (as we do in
    \cref{sec:reconstruct}), we interweave computing sweeping orders with
    reconstructing $i$-skeletons, and do not require any initial knowledge of
    $K$, other than its vertex set.

    \section{Connections to Directional Transforms}
    \label{append:discretization}
        The persistent homology transform (PHT) and Euler characteristic transform
        (ECT), first defined in~\cite{turner2014persistent}, map a geometric
        simplicial complex $K$ to the set of persistence diagrams or Euler
        characteristic functions, respectively, corresponding to lower-star/sublevel
        set filtrations of $K$ in each direction. In~\cite{turner2014persistent,
        ghrist2018euler}, we see that such sets (parameterized by the sphere of
        directions) uniquely correspond to $K$; that is, they are \emph{faithful}.
        In~\cite{ji2024injectivity}, we see conditions for the faithfulness of the
        \emph{quadratic ECT}, which replaces the hyperplane sweeps used in sublevel
        set filtrations with quadratic hypersurfaces. Directional transforms can be
        useful in shape classification applications, see,
        e.g.,~\cite{crawford2019predicting, turner2014persistent, wect2020,
        betthauser2018topological, hofer2017constructing, maria2020intrinsic,
        amezquita2022measuring}. 
        Such applications necessarily use a
        \emph{discretization} of the directional transform, i.e., a finite set of
        descriptors corresponding to a finite sample of the parameterizing
        directions.  
        Finding finite sets of directional topological descriptors from which a
        simplicial complex can still be reconstructed has been explored in several
        contexts~\cite{belton2018learning,belton2019reconstructing,betthauser2018topological,
        curry2022many,micka2020searching, graphsearch, fasy2024faithful}. In
        particular, this is the perspective taken by~\cite{graphsearch}; since
        indegree is computable using information contained in verbose persistence
        diagrams, the sweep algorithm for edge reconstruction produces a set of
        directions for which the corresponding discretization of the verbose PHT is
        faithful.

    Verbose topological descriptors are, roughly speaking, a record of how
    topological invariants change during a filtration, where we note the effects of
    adding a single simplex at a time, even if multiple simplices are added at
    a single parameter value. This varies from the more traditional
    \emph{concise} topological descriptors, which only consider the coarse
    information of how the topological invariants change for each subcomplex of the
    filtration. 
    Note that for an \emph{index filtration} (where each simplex appears
    at a distinct parameter value), verbose and concise topological descriptors
    are equivalent.
    Also note that verbose and concise are sometimes called augmented and non-augmented, respectively.

    If we choose to record homology, Betti number, or Euler characteristic, the resulting topological descriptor is the (verbose/concise) persistence diagram, Betti function, or Euler characteristic function, respectively.
    While precise definitions of these functions are not necessary for our discussion, we refer the reader to~\cite[Section 3 and Appendix A]{ordering} for a more careful treatment. 
    
    \subsection{Indegree from Verbose Persistence Diagrams and Betti Functions}
    \label{append:indegree}
     We note that the
        definition of indegree given in~\cite[Def. 3]{graphsearch} is restricted to
        vertices and edges (a special case of our definition), and the definition of
        $k$-indegree given in~\cite[Def. 21]{fasy2024faithful} allows us to choose the
        dimension of cofaces to count (a generalization of our definition).
    
    In the case of verbose persistence diagrams and (with trivial adaptations)
    verbose Betti functions, we can immediately determine the number of
    $(i+1)$-simplices that appear at the height of $\sigma$. From a single diagram
    corresponding to the direction in question, we simply count the number of
    $i$-dimensional points whose birth height is the same as the height of~$\sigma$
    plus the number of $(i+1)$-dimensional deaths whose birth height is the height
    of $\sigma$; this total is the number of $(i+1)$-simplices that appear at the
    height of $\sigma$ (see, e.g.,~\cite[Lemma 11]{fasy2024faithful}). However, this
    count may be more than the actual indegree of $\sigma$, since $(i+1)$-simplices
    that are not cofacets of $\sigma$ may have their highest vertex at the height of
    $\sigma$.
    
     In~\cite[Alg. 1]{fasy2024faithful}, such extra simplices that are ``attached''
     to a proper face of $\sigma$ are eliminated from the count by recursively
     finding the number of $(i+1)$-simplices incident to faces of $\sigma$,
     beginning with $\sigma$ itself, then with $(i-1)$-dimensional faces of
     $\sigma$, then with $(i-2)$-dimensional faces, and so on. Through an
     inclusion-exclusion computation, these totals are added and subtracted to
     obtain the true indegree of $\sigma$. This operation is limited to computing
     indegree in directions $s$ that isolate $\sigma$ in a hyperplane normal to $s$,
     and also requires various tilts to isolate subsets of vertices. Thus, for this
     implementation of indegree, we require the following general position
     assumption.
    
    \begin{assumption}[General Position for Indegree from Verbose Descriptors]
        Let $K$ be a simplicial complex in $\R^d$. To compute indegree using verbose
        persistence diagrams (or verbose Betti functions) using the methods
        of~\cite{fasy2024faithful}, we require every collection of $i$ vertices for
        $1 \leq i < d$ to be affinely independent. \label{ass:indegree}
    \end{assumption}
    
    The specifics of actually computing indegree using verbose persistence diagrams
    or verbose Betti functions mean that, in practice, computing indegree in a
    single direction requires verbose descriptors for filtrations from multiple
    directions. 
    
    \begin{lemma}[Complexity of Computing Indegree]
    Let $\sigma \in K$ be an $i$-simplex, and suppose that $s \in \sph^{d-1}$ is
    perpendicular to~$\sigma$, and no other vertices of $K$ have the same
    $s$-coordinate as~$\sigma$. Then the indegree of~$\sigma$ with respect to $s$
    can be computed using $2^{i+1}-1$ verbose persistence diagrams or verbose Betti functions.
    \end{lemma}
    
    The proof combines ideas from Algorithm 1 and Theorem 39
    of~\cite{fasy2024faithful}.

    \subsection{Challenges Computing Indegree with other Descriptor Types}
    
        As detailed in~\cite{fasy2024faithful},
        if we choose to use verbose Euler characteristic functions rather than
        verbose persistence diagrams or verbose Betti functions, we can no longer
        compute indegree. Instead, we can only compute \emph{even-/odd-degree}, or
        the number of even- (odd-) dimensional cofaces of and ``below'' a given
        simplex. While mild adaptations to the reconstruction algorithm
        of~\cite{fasy2024faithful} allows for this alternate type of query in the
        reconstruction process, our radial search algorithm (\cref{alg:findunfound})
        unavoidably relies on knowing the dimension of a coface. Namely,
        if the number of, e.g., even-dimensional cofaces of a given $i$-simplex $\sigma$ in
        some query direction matches the number we would expect if all candidates were cofacets of~$\sigma$, we still cannot guarantee they all are
        truly cofacets of $\sigma$, because lower/higher-dimensional cofaces may
        also contribute to this count.
    
    We are not aware of any methods to compute indegree using concise descriptors, other than reconstructing the entirety of $K$ using other methods, and using this information to report indegree (which rather defeats the purpose). 
    Development of such a method seems unlikely. 
    In various senses, concise descriptors have been shown to be weaker than their verbose counterparts~\cite{ordering, zhou2023beyond, ephemeral}.
    In particular, concise descriptors are
    sensitive to the ``flatness'' of a simplicial complex, or how close
    to affinely-dependent any subset of vertices is. Existing bounds on the size of faithful discretizations of concise directional transforms bound this flatness~\cite{curry2022many}, and without such restrictions on flatness, there are strict requirements on what directions are necessary~\cite[Corollary 1]{ordering}.
     
    \subsection{Connecting Reconstruction Algorithms to Faithful Discretizations}
    \label{append:faithful}
    If a set of topological descriptors can be used to reconstruct a simplicial complex, it is a faithful set. 
    Therefore, if we consider the set of topological descriptors that would be used to perform all our indegree queries (along with a set to reconstruct $K_0$), we would arrive at a faithful discretization of the associated directional transform, i.e., either the verbose persistent homology transform or the verbose Betti function transform.

    In \cref{sec:reconstruct}, we take a similar  perspective as~\cite{graphsearch,
    belton2019reconstructing}; namely, given no initial information about a
    simplicial complex $K$ other
    than general position and possibly $K_0$, we then aim to reconstruct $K$ using indegree queries.
    This perspective is fundamentally different from
    ~\cite{fasy2024faithful}, which, using \emph{total} information about $K$,
    defines a set of directions from which $K$ could be reconstructed (i.e., corresponding to a
    faithful discretization of the associated topological transform).
    Using initial knowledge of $K$
    to define the direction set
    allows~\cite{fasy2024faithful} to lessen the impact of unnecessary
    ``exploratory'' queries. In particular, when reconstructing $K_{i+1}$ given
    $K_i$, the reconstruction algorithm of~\cite{fasy2024faithful} certifies the
    presence of each $(i+1)$-simplex, and since the total number of
    $(i+1)$-simplices can be read off any single diagram, there is no need to test
    remaining candidates that are not true $(i+1)$-simplices.

    Since we assume no initial knowledge of $K$, unlike~\cite{fasy2024faithful}, we don't know
    where to look in order to certify only the true $(i+1)$-simplices; we inevitably
    encounter and test candidate~$(i+1)$-simplices that are not in~$K_{i+1}$. While
    this results in ``extra'' queries from the perspective of purely building a
    finite faithful set of directions, it is necessary for our framework.
\section{Proof of Lemma \ref{lem:orderingordering}}
\label{append:proofs}
\label{append:orderingordering}
To prove \cref{lem:orderingordering}, we begin by establishing two helpful
lemmas. The first serves as our eventual base case.

\begin{lemma}[Candidate-Ordering for All Vertices]
\label{lem:vertexordering}
    Let $K$ be a simplicial complex in $\R^d$. Then there exists a circle of directions that is candidate-ordering for every vertex of $K_0$.
\end{lemma}
\begin{proof}
     First, consider the case $d=2$.
     By \cref{lem:order}, each vertex $v \in K_0$ has some candidate-ordering circle.
     Then the unique circle of directions is candidate-ordering, and the
     claim is satisfied trivially. Suppose then that $d>2$.
     Each triple of affinely independent vertices defines a plane; let $S
     \subset \sph^{d-1}$ be the set of directions normal to such a plane, for
     some triple in $K_0$. Since $K_0$ is finite, so is $S$, and we can find
     some circle $\gamma$ such that $\gamma \cap S = \emptyset$. We claim
     $\gamma$ is candidate-ordering for all vertices.

     Consider some $v \in K_0$, and suppose, towards a contradiction, that
     $\gamma$ is \emph{not} candidate-ordering for $v$. That is, when rotating
     around $v$ by the circle $\gamma$, there are two candidate vertices of $v$,
     which we call $v_1$ and $v_2$, that appear at the same ``angle'' around
     $\gamma$. In particular, this means that, for some $s \in \gamma$, we have
     $s \cdot v = s\cdot v_1 = s \cdot v_2$, but $\aff(v, v_1, v_2)$ is
     two-dimensional. But then this $s$ is normal to $\aff(v, v_1, v_2)$, i.e.,
     $s \in S$, which contradicts our assumption that $\gamma$ is disjoint from~$s$.
\end{proof}

Next, we observe that the candidate vertices of a simplex are also candidate vertices for a facet of that simplex.

\begin{lemma}
\label{lem:sharedcandidates}
    Let $K$ be a simplicial complex in $\R^d$, and, for $0< i< d-1$ consider
    an~$(i-1)$-simplex $\rho \in K$. Let $\sigma$ be a cofacet of $\rho$. Then,
    if $v \in K_0$ is a candidate vertex of~$\sigma$, it is also a candidate
    vertex of $\rho$.
\end{lemma}
\begin{proof}
   Recall from \cref{def:candidate} that the simplices of $K_i$ defined on $\sigma
   \cup \{v\}$ form the boundary of an $(i+1)$-simplex $\tau$ such that $K_i
   \cup \tau$ is subcomplex of some simplicial complex $K'$ that satisfies
   \cref{ass:reconstruction}. Note that the simplices defined on $\sigma \cup
   \{v\}$ are a superset of the simplices defined on $\rho \cup \{v\}$; in
   particular, simplices on $\rho \cup \{v\}$ in $K_{i-1}$ form the boundary of
   an $i$-simplex $\sigma'$. Then $K_{i-1} \cup \{\sigma'\}$ is also a
   subcomplex of $K'$, which we know satisfies \cref{ass:reconstruction}.
   Therefore, $v$ is also a candidate vertex for $\rho$.
\end{proof}

Finally, we are ready to prove \cref{lem:orderingordering}, which we restate below.

\orderingordering*

\begin{proof}
    We claim that we can construct such a sweeping order by iteratively calling
    \cref{alg:sweep} and specifically using candidate-ordering circles
    on \alglnref{sweep:gamma} for each call.
    We proceed by induction on $i$. 
    Consider the base case $i=0$.
    By \cref{lem:vertexordering} there exists some circle of directions that is candidate-ordering for \emph{all} vertices. 
    Choosing some arbitrary direction~$s$ from this circle on \alglnref{sweep:vertdir} in \cref{alg:sweep} results in the sweeping order $((v_j, s))_{j=1}^{n_0}$, which satisfies the claim.
    
    Next, for some $\ell-1\geq 0$, suppose that~$((\rho_{j},
    s_{j}))_{j=1}^{n_{\ell-1}}$ is a sweeping order for $K_{\ell-1}$, where
    every $s_j$ is part of a candidate-ordering circle $\gamma_j$ around
    $\rho_j$, and where we specifically rotate around $\gamma_j$ in
    \alglnref{sweep:gamma} to order cofacets of $\rho_j$.
    Suppose we compute a sweeping order for $K_\ell$ by using $((\rho_{j},
    s_{j}))_{j=1}^{n_{\ell-1}}$ as input to \cref{alg:sweep}. Consider $(\sigma,
    s)$, an arbitrary term of the output, and let $\rho_k$ denote the first
    $(\ell-1)$-simplex in the sweeping order for $K_{\ell-1}$ that contains
    $\sigma$ as a cofacet. That is, $(\rho_k, s_k)$ is the pair considered on
    \alglnref{sweep:forbegin} of \cref{alg:sweep} that corresponds to identifying
    $\sigma$. Furthermore, $\gamma_k$ used on \alglnref{sweep:gamma} is candidate-ordering for $\rho_k$ and $s \in \gamma_k$.
    
    Now, suppose, towards a contradiction, that $s$ is not part of any candidate-ordering circle for $\sigma$. In particular, this means that $\sigma$ has two candidate vertices, $v_1$ and $v_2$, such that the $s$-coordinates of $\sigma, v_1,$ and $v_2$ agree, but such that $\aff(\sigma \cup \{v_1\}) \neq \aff(\sigma \cup \{v_2\}$), i.e., $v_1$ and $v_2$ would be assigned the same angle around a circle containing $s$.
    However, by \cref{lem:sharedcandidates}, we know that $v_1$ and $v_2$ are
    also candidate vertices for $\rho_k$. Furthermore, since $\rho_k$ is a
    facet of $\sigma$, the $s$-coordinates of $\rho_k$, $v_1$, and $v_2$ agree,
    and $\aff(\rho \cup \{v_1\}) \neq \aff(\rho\cup \{v_2\})$. This contradicts
    our assumption that $\gamma_k$ was candidate-ordering for $\rho_k$.
    Then, by induction, we have shown the claim.
\end{proof}

\end{document}